\documentclass[a4paper,fleqn]{cas-sc}

\usepackage[authoryear,longnamesfirst]{natbib}
\usepackage{cleveref}
\usepackage{algorithm}
\usepackage{algorithmic}
\usepackage{enumerate}
\usepackage{enumitem}
\usepackage{subfig}

\newcommand{\bx}{\mathbf{x}}
\newcommand{\bn}{\mathbf{n}}
\newcommand{\by}{\mathbf{y}}
\newcommand{\bi}{\mathbf{i}}
\newcommand{\Ht}{H^{\tau}}

\newcommand{\Mat}[1]{\underline{\mathbf {#1}}}

\def\tsc#1{\csdef{#1}{\textsc{\lowercase{#1}}\xspace}}
\tsc{WGM}
\tsc{QE}
\newtheorem{theorem}{Theorem}
\newtheorem{lemma}[theorem]{Lemma}
\newtheorem{corollary}[theorem]{Corollary}
\newdefinition{rmk}{Remark}
\newproof{pf}{Proof}
\begin{document}
\let\WriteBookmarks\relax
\def\floatpagepagefraction{1}
\def\textpagefraction{.001}

\shorttitle{Estimation and UQ for Seabed and Its Roughness}    

\shortauthors{}  

\title [mode = title]{Estimation and Uncertainty Quantification for Seabed and Its Roughness With Longitudinal Waves}  

\tnotemark[1]

\tnotetext[1]{This work was funded by Villum Investigator grant (no.\ 25893) from the Villum Foundation, partially supported by Research Council of Finland project numbers 359186 and 353093, and by the FEDER/MICINN--AEI under grants no.\ PID2020-112796RB-C21 and no.\ PID2024-155528OB-C21.}

\author[1]{Babak Maboudi Afkham}

\cormark[1]

\ead{babak.maboudi@oulu.fi}

\credit{Conceptualization, Methodology, Software, Formal analysis, Investigation, Visualization, Writing -- original draft, Writing -- review \& editing}

\affiliation[1]{organization={Research Unit of Mathematical Sciences, University of Oulu},
            addressline={Pentti Kaiteran katu 1, Linnanmaa}, 
            city={Oulu},
            country={Finland}}

\author[2]{Ana Carpio}

% Email id of the second author
\ead{ana_carpio@mat.ucm.es}

\credit{Conceptualization, Methodology, Formal analysis, Investigation, Writing -- review \& editing, Supervision}

\affiliation[2]{organization={Departamento de Matem\'atica Aplicada, Universidad Complutense de Madrid},
            city={Madrid},
            postcode={28040},
            country={Spain}}

\cortext[1]{Corresponding author}

\begin{abstract}
This paper introduces an infinite-dimensional Bayesian framework for acoustic seabed tomography, leveraging wave scattering to simultaneously estimate the seabed and its roughness. Tomography is considered an ill-posed problem where multiple seabed configurations can result in similar measurement patterns. We propose a novel approach focusing on the statistical isotropy of the seabed. Utilizing fractional differentiability to identify seabed roughness, the paper presents a robust numerical algorithm to estimate the seabed and quantify uncertainties. Extensive numerical experiments validate the effectiveness of this method, offering a promising avenue for large-scale seabed exploration.
\end{abstract}


\begin{highlights}
\item Introduced hierarchical Bayesian inference for seabed roughness estimation
\item Seabed roughness identified via fractional differentiability of random fields
\item Posterior distribution proven well-posed in infinite dimensions
\item Developed sampling algorithms for nonsmooth inverse scattering posteriors
\item Accurate uncertainty-aware reconstruction demonstrated numerically
\end{highlights}

% Keywords

\maketitle

% Main text
%\section{}\label{}

% Use the style of numbering in square brackets.
% If nothing is used, default style will be taken.
%\begin{enumerate}[a)]
%\item 
%\item 
%\item 
%\end{enumerate}  

%\begin{itemize}
%\item 
%\item 
%\item 
%\end{itemize}  

%\begin{description}
%\item[]
%\item[] 
%\item[] 
%\end{description}  

% Uncomment and use as the case may be
%\begin{theorem} 
%\end{theorem}

% Uncomment and use as the case may be
%\begin{lemma} 
%\end{lemma}


\section{Introduction} \label{sec:intro}

Understanding the seabed is crucial for many offshore activities, particularly in renewable energy development and marine transportation. Accurate inference of seabed topographic characteristics, such as smoothness versus roughness, remains a challenging problem with significant practical implications \cite{WANG2025104383,anderson2008acoustic}. A notable example is acoustic backscatter ambiguity \cite{lurton2015backscatter,clarke2015multispectral}, where active sonar systems struggle to distinguish between smooth, hard seabeds and rough, soft ones, leading to misclassification. Another important challenge is acoustic reverberation \cite{jia2021underwater,jenserud2015measurements}, in which multiply scattered and reflected waves contaminate measurements; mischaracterizing seabed roughness can substantially bias interpretation and degrade imaging performance. This work presents a framework for estimating seabed roughness, characterized by fractional differentiability, together with uncertainty quantification, contributing to more reliable acoustic seabed characterization.

Seabed inversion problems aim to infer seabed properties as an alternative to costly drilling-based exploration by analyzing how acoustic or elastic waves interact with the seabed \cite{jackson2007high,tromp2020earth}. Classical formulations seek to estimate parameters such as sound speed, density, and attenuation of seabed layers, and include Yamamoto theory inversion \cite{yamamoto}, waveguide characteristic impedance methods \cite{waveguide}, multistep optimization techniques \cite{multistep}, and Bayesian approaches \cite{bayesian1,bayesian2}. Much of the existing literature focuses on shallow-water environments and the spatial variability of soft sediments. More recently, seabed tomography techniques that exploit the full wavefield have been developed \cite{tomography}. In these approaches, waves emitted from sea-surface or subsurface sources propagate through the medium and are scattered by the seabed, with the resulting signals recorded by sensors \cite{anderson2008acoustic,tromp2020earth}. The geometry and physical characteristics of the seabed, including biological structures and rock formations, are encoded in these measurements \cite{frederick2020seabed}. The choice between acoustic and elastic modeling depends on seabed properties: elastic models provide greater fidelity for complex, layered, or solid seabeds through explicit representation of elastic moduli, whereas simpler acoustic models are appropriate for fluid-like seabeds but may introduce inaccuracies in representing solid structures.

In this paper, we focus on characterizing deep-water seabed profiles by Bayesian tomographic techniques using elastic waves. 
More precisely, we propose computational techniques to estimate the seabed roughness.
The seabed roughness, also referred to as topographic variability, is a crucial parameter in marine biology, geophysics, as well as ocean engineering, in relation with applications such as submarine cabling or stability of floating wind turbine platforms, for instance, see, e.g., \cite{anderson2008acoustic,bjorno2017applied,frederick2020seabed,valentine2005classification} and the references therein.
There have been many attempts to quantify the statistical properties of roughness, e.g., 
using spectral statistics, auto-correlation, and fractal analysis \cite{bjorno2017applied}. In the context of 
seabed tomography, miss-identifying the roughness characteristics can lead to artifacts, e.g., reverberation \cite{bjorno2017applied}, which compromise the quality of estimations.
A similar situation is observed in Biomedicine when inferring the shape of the boundary of tumors \cite{carpio2023shear}.

%The seabed roughness, also referred to as topographic variability, is a crucial parameter in marine biology, geophysics, as well as ocean engineering, see, e.g., \cite{anderson2008acoustic,bjorno2017applied,frederick2020seabed,valentine2005classification}. There have been many attempts to quantify the statistical properties of roughness, e.g., 
%using spectral statistics, auto-correlation, and fractal analysis \cite{bjorno2017applied}. In the context of 
%seabed tomography, miss-identifying the roughness characteristics can lead to artifacts, e.g., reverberation \cite{bjorno2017applied}, which compromise the quality of estimations.

The problem of inferring an object from indirect measurements is an inverse problem. Such problems are often ill-posed, due to data incompleteness and measurement noise  \cite{hansen2021computed,kaipio2006statistical}. Seabed tomography,
i.e., estimating the seabed profile from scattered wave patterns on the sea-surface,
is an example of an ill-posed inverse problem.
Only partial and noisy measurements of the scattered wave are possible. Therefore, different seabed configurations can produce similar measurements \cite{frederick2020seabed,borcea2023data}.

The Bayesian approach to inverse problems is a promising method to handle ill-posed inverse problems. In this approach, model parameters, unknowns, and noise are modelled as random variables or fields. When the ill-posedness of an inverse problem is a result of an incomplete measurement, we incorporate our knowledge of the solution in the form of a prior distribution. The solution to a Bayesian inverse problem is the posterior distribution, i.e., the conditional distribution of the unknown given the measurement and the prior distribution. Therefore, in the Bayesian setting, credible solutions are characterized by their posterior probability.

We consider the seabed curve to be a random function. A classic approach to incorporate distributions of random functions, e.g., using Markov random fields \cite{koller2009probabilistic,suuronen2022cauchy,bardsley2012laplace}, into the Bayesian framework is by defining probability distributions on the discrete random functions. In this approach, the number of problem-unknowns often scales with discretization level. This leaves such methods impractical for large-scale problems. Infinite-dimensional Bayesian methods, is a recent approach for incorporating distributions of random functions, also referred to as \emph{random fields}, into the posterior distribution. These methods decouple statistical and spatial discretization of random functions. Therefore, resulting in converging statistical solutions for discrete problems with a fine spatial discretization.

In this paper we present an infinite-dimensional Bayesian framework for the seabed tomography problem. We take a goal-oriented approach by defining the unknown variable of the model to be the seabed interface. A central hypothesis in this paper is that the statistical description of the the seabed is isotropic, i.e., the seabed is a stationary and translation-invariant random function. This is a common assumption, e.g., see Chapter 8 in \cite{bjorno2017applied}. The implication of this assumption is that we can infer global properties of the seabed from local measurements.

In the context of scattered light from a surface, identifying global roughness from local measurement is done in the works of \cite{karamehmedovic2013efficient,schroder2011modeling}.
These works estimate the power of each discrete Fourier mode of a surface,
%in terms of its power spectral density
from the scattered light by a small region of the surface. Although these methods are effective in identifying the surface roughness for some range of frequencies, they are not suitable for inferring the surface profile. In addition, quantifying uncertainties
via such methods, e.g., due to model errors or incomplete measurements, is still an open direction of research.

In this paper we identify the roughness of an isotropic seabed by means of its level of fractional differentiability. To our knowledge, this is a novel method for inverse scattering problems. Unlike the above method which considers a finite number of frequencies, the fractional differentiability considers the \emph{decay} of the power spectral density function (i.e., its asymptotic behaviour).
It is shown in \cite{maboudi2024} that identifying the roughness via the fractional differentiability is robust to noise and incomplete measurements, for a wide range of problems.

We define the posterior distribution to be the joint probability distribution of the seabed's profile and roughness,
given the measurements. We show that the joint posterior distribution is well-posed, i.e., it is locally Lipschitz with respect to the measurement. We present a numerical algorithm which explores the posterior, provides an estimate of the roughness and the seabed, and quantifies the uncertainties as well. We carry extensive numerical experiments to test the method. These experiments suggest that our method gives an accurate estimation of the seabed. In addition, analyzing the uncertainty reveals the regions where the estimated seabed diverges from the true seabed. Therefore, this method is a promising approach for large-scale exploration of the seabed. While the primary focus is on seabed estimation, the underlying framework is broadly applicable to a wide range of inverse scattering problems.

While the proposed model adopts simplifying assumptions such as isotropy and homogeneity of seabed properties, these choices are made deliberately to isolate and study the effect of seabed roughness within a principled Bayesian framework. The primary contribution of this work is a methodology that enables direct inference of seabed roughness through fractional regularity, without relying on post-processing of reconstructed fields. Extending the approach to real-world applications will require further investigation of noise sensitivity, more complex ocean–ground interactions, and fully three-dimensional settings.

The paper is organized as follows. We describe the interaction of the emitted waves with the seabed in terms of   wave equations in \Cref{sec:problem}. Furthermore, we suggest a numerical procedure to solve this problem and formulate the seabed tomography problem as an inverse problem. In \Cref{sec:bayes}, we take a goal-oriented approach and reformulate the seabed tomography problem in Bayesian terms. We construct the posterior distribution following a Bayesian approach and provide a numerical algorithm for exploring this posterior. Numerical experiments are provided in \Cref{sec:results}, where we test the method for prior samples and out-of-prior samples. We present conclusive remarks in \Cref{sec:conclusion}.

\section{Problem Formulation} \label{sec:problem}
In this section we introduce the problem of inferring the seabed profile from surface measurements and discuss the associated challenges. 
%We first introduce the classical two-step method to inferring the seabed, i.e., first inferring the density and %elastic coefficients of the wave equation, and then inferring the seabed. 
We first introduce the model governing the interaction of the emitted waves with the seabed (the so-called forward model) and define the observation operator that relates the model to the data. For this, we choose a simplified wave equation to represent the propagation of longitudinal elastic waves. Then, we introduce the goal-oriented approach of directly inferring the seabed interface from data.

\subsection{The Forward Model:\ Scalar Wave Equation}
The propagation of longitudinal elastic waves
(sometimes referred to as p-waves)
emitted by a source can be modelled with the scalar wave equation. Let $\Omega\subset \mathbb R^2$ be a rectangular domain.
We use the following notation:\
$\bx=(x,y)\in \Omega$ are Cartesian coordinates with $x\in[-3,3]$ and $y\in [-1.5,1.5]$, $u:\mathbb R^2\times \mathbb R\to \mathbb R$ is the displacement, and $\rho, \alpha:\mathbb R^2\to \mathbb R$ are positive density and elastic coefficients, respectively.
Moreover, inspired by \cite{carpio2023shear,tsogka2002time}, we introduce
dimensionless source terms $f_i(t)g(\bx - \bx_{s_j})$ that
model point sources at locations $\bx_{s_j}$, $j=1,\dots,N_s$ with 
\begin{equation} \label{eq:freqs}
    f_i(t) = (1-2\pi^2(f^0_i)^2t^2)e^{-\pi^2(f_i^0)^2t^2}, \qquad i=1,\dots,N_f. 
\end{equation}
where $f^0_i>0$, $i=1,\dots,N_f$, are the central frequencies of the sources. Then the forward problem takes the form,
for $i=1,\dots,N_f, ~ j = 1,\dots,N_s$:
\begin{equation} \label{eq:wave-eq}
\left\{
    \begin{aligned}
        & \rho(\bx) u_{tt}(\bx,t) - \nabla \cdot( \alpha(\bx) \nabla u(\bx,t) )=  \qquad & \bx \in \Omega, t>0, \\
        & \qquad f_i(t)g(\bx - \bx_{s_j}), \\
        & \frac{\partial u(\bx,t)}{\partial \bn} = 0, & \bx \in \partial \Omega^+, t>0, \\
        & \frac{\partial u(\bx,t)}{\partial \bn} = -\frac{ u_t(\bx,t)}{c(\bx)},  & \bx \in \partial \Omega^-, t>0, \\
        & u(\bx,0) = u_t(\bx,0) \equiv 0,  & \bx \in \Omega. \\
    \end{aligned} 
\right.
\end{equation}
Here, $g(\mathbf s) = \exp(-\|\mathbf s\|^2/0.0025)$, with $\mathbf s=(s_1,s_2)$ and sources are uniformly distributed in the $x$-axis at the depth (y component) of $1.4$. The model has been non-dimensionalized using 1km and 1s as space and time reference magnitudes. Note that $c:\mathbb R^2 \to \mathbb R^+$, in \eqref{eq:wave-eq}, is the wave speed obtained from the relation $\alpha = \rho c^2$.

%The dimensionless source term, $f_i(t)g(\bx - \bx_{s_j})$, on the right-hand-side of \eqref{eq:wave-eq}, 
%models a point sources at locations $\bx_{s_j}$, $j=1,\dots,N_s$, and are inspired by \cite{tsogka2002time,carpio2023shear}, with 
%\begin{equation} \label{eq:freqs}
%    f_i(t) = (1-2\pi^2(f^0_i)^2t^2)e^{-\pi^2(f_i^0)^2t^2}, \qquad i=1,\dots,N_f. 
%\end{equation}
%where $f^0_i>0$, $i=1,\dots,N_f$, are the central frequencies of the sources. 

The top boundary of the domain, for the ocean surface,
is denote by $\partial \Omega^+$, while the rest is denoted by $\partial \Omega^-$, in such a way that $\partial \Omega^+\cup \partial \Omega^- = \partial \Omega$.
We impose a homogeneous Neumann boundary condition on $\partial \Omega^+$ and a non-reflective boundary condition for the rest of the boundaries to model the open sea condition. We define the solution operator $\mathcal W$ to be the mapping $(\rho, \alpha) \mapsto u$.

We collect  measurements from the ocean surface and model the measurement process with a set of observation operators $\mathcal O_i$ which for a given solution $u$, associated with coefficients 
$(\rho, \alpha)$, provides the discrete time-series of ocean surface  measurements
\begin{equation}
    \mathcal O_i(u) = u(\bx^{\text{sensor}}_k, t_{l}), \qquad k=1,\dots,N_{\text{sensor}}, ~ l = 1,\dots, N_{t},
\end{equation}
for $i=1,\dots,N_f$ frequencies, where $\bx^{\text{sensor}}_k$ stands for the location of $N_{\text{sensor}}$ ocean surface sensors whose exact geometry is discussed below. Note that the time period of interest is $t\in [0, t_{\text{max}}]$ with $t_{\text{max}} = t_{N_t}$. The measurement is often corrupted by noise. We consider a Gaussian additive noise and define the measurement to be
\begin{equation} \label{eq:inverse}
    \by^{\text{obs}}_i= \mathcal G_i(\rho, \alpha) + \varepsilon_i I, \qquad i=1,\dots,N_f,
\end{equation}
where $\mathcal G$ is the forward operator defined as $\mathcal G :=\mathcal O \circ \mathcal W$, and $\varepsilon\sim \mathcal N(0, \sigma^{\text{noise}}_i I)$,  $I$ being the identity matrix of size $N_{\text{sensor}}\times N_{t}$. The seabed tomography problem aims to identify the seabed, i.e., the interface between rock and water, from noisy measurements $\by^{\text{obs}}_i$, $i=1,\dots,N_f$.

It is sometimes useful to investigate time-snapshots of the measurement. To this end, we introduce the slicing notation $\by^{\text{obs}}_i[l]$ and $\mathcal G_i(\rho, \alpha)[l]$ which referes to the $l$th column of $\by^{\text{obs}}_i$ and the $l$th solution snapshot $u(\mathbf x,t_l)|_{\partial \Omega ^+}$, respectively.

We remark that there are alternative ways to model wave propagation, e.g., using the elastodynamics \cite{tsogka2002time} while our choice is due to its simplicity and versatility.

%\begin{remark}
%    The forward operator $\mathcal G$ is a nonlinear and non-differentiable forward operator due to the non-linearity in $\rho$ and $\alpha$ with respect to $h$.
%\end{remark}

The seabed tomography problem is inferring $(\rho,\alpha)$ from the measurement time sequences
$\by^{\text{obs}}_i$, for $i=1,\dots, N_f$. Since data are only available at the ocean surface, the inverse problem presented in \eqref{eq:inverse} is considered to have missing data. This makes the inverse problem severely ill-posed, i.e., it may lack a unique solution \cite{borcea2021reduced}.

Direct discretization of this problem with space dependent coefficients $\rho$ and $\alpha$ becomes rather costly. One particular challenge for discretizing a large-scale system is that a discrete description of $(\rho, \alpha)$, in terms of their values at each mesh point, scales with the discretization level. This results in discrete matrices with large condition numbers, small time steps for stability and unaffordable computational costs. Instead, in the following section we exploit the piece-wise constant nature of $(\rho, \alpha)$ to develop a discretization invariant inverse problem.

\subsection{Goal-Oriented Inverse Problem}
In this section we assume that $\rho$ and $\alpha$ are piece-wise smooth
in the water and piecewise constant in the sub-seabed domain, and that water and sub-seabed have known constant density and elastic coefficients. Here, we work in a layered geometry, i.e., each layer corresponds to a different material and their density and velocity profiles are known. Generally, the properties of water (e.g., density and elastic modulus) vary spatially due to changes in pressure, temperature, and salinity. In this study, depth dependence is incorporated through a linearized approximation of the seawater equation of state, representing a first-order fit to the pressure dependence of density and bulk modulus \cite{fofonoff1983algorithms}. This approximation is intended for simplicity and analytical tractability; more complex and fully nonlinear depth-dependent relations may be incorporated in the same framework if required. For simplicity, the density and elastic properties of the rock are assumed to be constant, although spatially varying rock properties and multi-layered rock structure could also be modeled in a similar manner. Therefore, we may identify $\rho$ and $\alpha$ by the interface function $h:\mathbb R \to \mathbb R$ defining the seabed profile
\begin{equation}
\rho((x,y); h) = \left\{
    \begin{aligned}
        &\rho^+(y), \qquad &y > h(x), \\
        &\rho^-, \qquad &y \leq h(x),
    \end{aligned} 
\right.
\quad
\alpha ((x,y); h) = \left\{
    \begin{aligned}
        &\lambda^+(y), \qquad &y > h(x), \\
        &\lambda^- + 2\mu^-, \qquad &y \leq h(x),
    \end{aligned} 
\right. \label{coefficients}
\end{equation}
where $\rho^+(y)= \rho_0(1 + 0.0044(1.5-y)), \lambda^+=\lambda_0(1+0.025(1.5-y))$  and $\rho^-, \lambda^-, \mu^-$ are density and Lam\'e moduli \cite{landau2012theory} of water and 
sub-seabed rock, respectively. In our simulations, we will typically use  material constants
 $\rho^-, \lambda^-, \mu^-$ characterizing basalt, as an example of a typical material of the seabed. The material parameters used in the numerical experiments correspond to a hard-rock seabed, which typically exhibits strong acoustic contrast and increased roughness relative to the overlying water column, whereas softer sedimentary seabeds tend to be smoother and display weaker contrasts \cite{jensen2011computational}.

We now introduce the goal-oriented forward operator $\mathcal G_i^{\text{goal}}$ to be the map $h \mapsto u$. The measurement model now takes the form
\begin{equation} \label{eq:inverse-goal}
    \by^{\text{obs}}_i= \mathcal G^{\text{goal}}_i(h) + \varepsilon_i, \qquad i=1,\dots,N_f.
\end{equation}
We remark that to eliminate computational artifacts of the boundaries, we require $h$ to meet the boundaries of the domain $\Omega$ in an orthogonal direction, i.e., the tangent to $h$ is normal to the boundary, at the location of the boundary.

\subsection{Finite Element Discretization} In this section we discuss the numerical approximation of $\mathcal W$, i.e., the approximate solution to \eqref{eq:wave-eq}. We use the finite-element method (FEM) to discretize
\eqref{eq:wave-eq} in space and finite differences (FD) in time. We provide the fully discrete solver at the end of this section. Let us introduce  the velocity $v := u_t$. This transforms \eqref{eq:wave-eq} to 
\begin{equation} \label{eq:wave-eq-1st-order}
\left\{
    \begin{aligned}
        & \rho(\bx; h) v_{t}(\bx,t; h) - \nabla \cdot( \alpha(\bx; h) \nabla u(\bx,t; h) ) -  \qquad & \bx \in \Omega, t>0, \\
        & \qquad f_i(t)g(\bx - \bx_{s_j}) = 0, \\
        & u_{t}(\bx,t; h) - v(\bx,t; h) = 0, & \bx \in \Omega, t>0, \\
        & \frac{\partial u(\bx,t; h)}{\partial \bn} = 0, & \bx \in \partial \Omega^+, t>0, \\
        & \frac{\partial u(\bx,t; h)}{\partial \bn} + \frac{ v(\bx,t; h)}{c(\bx; h)}=0,  & \bx \in \partial \Omega^-, t>0, \\
        & u(\bx,0; h) = v(\bx,0; h) \equiv 0,  & \bx \in \Omega. \\
    \end{aligned} 
\right.
\end{equation}
Let $w\in H^1(\Omega)$ and $q\in L^2(\Omega)$ \cite{Quarteroni2014} be test functions. Multiplying  \eqref{eq:wave-eq-1st-order} by the test functions, using Green's identity  and applying the boundary conditions we find
\begin{equation} \label{eq:weak-form}
\left\{
    \begin{aligned}
        &\int_\Omega \rho \, v_{t} w  \ d\mathbf x 
         + \int_\Omega \rho c^2  \, \nabla u \cdot \nabla w \ d\mathbf x
         + \int_{\partial \Omega^-}   \rho c \,  v w \ d s  
         - \int_\Omega   f_i g  w \ d\mathbf x = 0, \\
         &\int_\Omega u_t q \ d\mathbf x- \int_\Omega v q  \ d\mathbf x = 0.
    \end{aligned}
\right.
\end{equation}
Discretizing this equations on a uniform FEM mesh and applying the method of lines \cite{Edsberg2015-rk} yields the semi-discrete system of equations
\begin{equation} \label{eq:semi-discrete}
    \begin{aligned}
        &\Mat{R}  \frac{d}{dt} \mathbf v 
        + \Mat{D} \mathbf u +   \Mat{C}  \mathbf v - f_i(t) \mathbf b = 0 \\
        & \Mat{M}_1 \frac{d}{dt} \mathbf u - \Mat{M}_2 \mathbf v = 0,
    \end{aligned}
\end{equation}
where, $i = 1,\dots, N_f$. The components of matrices $\Mat R, \Mat D, \Mat C, \Mat M_1$ and $\Mat M_2$ are given by
\begin{equation} \label{eq:FEM_matrices}
    \begin{aligned}
     \relax[ \Mat{R} ]_{mn} &= \int_\Omega  \rho \, \phi_m \phi_n \ d\mathbf x, \qquad
     &[ \Mat{D} ]_{mn} &= \int_\Omega  \rho c^2 \, \nabla \phi_m \cdot \nabla \psi_n \ d\mathbf x,  \\
    [ \Mat C ]_{mn} &= \int_{\Gamma^-} \rho c \,  \phi_m \phi_n \ d s,
    &[ \mathbf b ]_{m} &=  \int_\Omega  g  \phi_m  \ d\mathbf x, \qquad \\
    %[ \mathbf b ]_{n} &=  \int_\Omega  g  \phi_k  \ d\mathbf x \\,
    [ \Mat M_1 ]_{mn} &=  \int_\Omega \psi_m \psi_n \ d\mathbf x,
    &[ \Mat M_2 ]_{mn} &=  \int_\Omega \psi_m \phi_n \ d\mathbf x, \qquad \\
    \end{aligned}
\end{equation}
for, $1\leq m,n \leq N_{\text{FEM}}$ where $[\cdot ]_{mn}$ represents the element of the matrix located at the $m$th row and $n$-th column, and $\{ \phi_1, \dots, \phi_{N_{\text{FEM}}}\}$ and $\{ \psi_1, \dots, \psi_{N_{\text{FEM}}}\}$ are first-order Lagrange polynomials approximating $W^{N_{\text{FEM}}}$ of $H^1(\Omega)$ and $Q^{N_{\text{FEM}}}$ of $L^2(\Omega)$ \cite{Quarteroni2014}, respectively. Furthermore, $\mathbf u$ and $\mathbf v$, in \eqref{eq:semi-discrete}, are the vectors of coefficients of $u$ and $v$ in these bases, respectively. Without loss of generality, we consider the spaces in which we expand $u$ and $v$ to be identical, i.e., $\phi_i=\psi_i$, for all $i$. This assumption makes all the matrices in \eqref{eq:FEM_matrices} symmetric and reduces the second equation in \eqref{eq:semi-discrete} to $ d\mathbf u/dt - \mathbf v = 0$. The vector $f_i(t)\mathbf b$ is the discrete source term associated with the $i$th frequency.

The choice of the time-integration method is crucial in effective seabed tomography and roughness estimation. Numerical dissipation or dispersion phase error introduced by the time integrator can bias the effective frequency response for the forward operator, especially at higher frequencies that carry most of the resolution. A wave equation, as formulated above, admits a Hamiltonian structure inside the domain. We therefor choose the second order symplectic St\"ormer-Verlet time-stepping scheme \cite{hairer2006structure} to preserve this Hamiltonian structure and avoids systematic numerical energy drift. Symplectic schemes exhibit bounded long-time energy error and do not introduce artificial damping of high-frequency modes in the conservative setting. The fully discrete model is given by
\begin{equation} \label{eq:fully-discrete}
    \begin{aligned}
        &\Mat R \frac{\mathbf v^{n+1/2} - \mathbf v^{n}}{\Delta t} +   \Mat D \mathbf u^{n} +  \Mat C  \mathbf v^{n} - f_i(t^n) \mathbf b  = 0, \\
        &\Mat M \frac{\mathbf u^{n+1} - \mathbf u^{n}}{\Delta t} - \Mat M \mathbf v^{n+1/2} = 0,\\
        &\Mat  R \frac{\mathbf v^{n+1} - \mathbf v^{n+1/2}}{\Delta t} +  \Mat D \mathbf u^{n+1} +   \Mat C  \mathbf v^{n+1/2} - f_i(t^{n+1/2}) \mathbf b = 0. \\
    \end{aligned}
\end{equation}

Evaluating the inverse of $\Mat R$ and $\Mat M$, e.g., to solve for $\mathbf v^{n+1/2}$ and $\mathbf v^{n+1}$ is computationally demanding for large-scale problems. Therefore, We use  conjugate-gradient iterations \cite{gatto2022mathematical} to obtain $\mathbf v^{n+1/2}$ and $\mathbf v^{n+1}$ from \eqref{eq:fully-discrete}.

In practice, for each source located at $\bx_j$, $j=1,\dots,N_s$ the corresponding  wave may be measured individually. However, the solution to the simultaneous emission of the source terms is the superposition of the individual solutions. Therefore, we perform simultaneous emissions in this paper.

Time snapshots of the solution $u$ to the discrete model \eqref{eq:fully-discrete}  for the emission of 5 sources with equal frequencies and some seabed $h$ can be found in \Cref{fig:snapshots}. Details of the implementation can be found below in \Cref{sec:results}.

\begin{figure}
    \centering
    \includegraphics[width=0.8\textwidth]{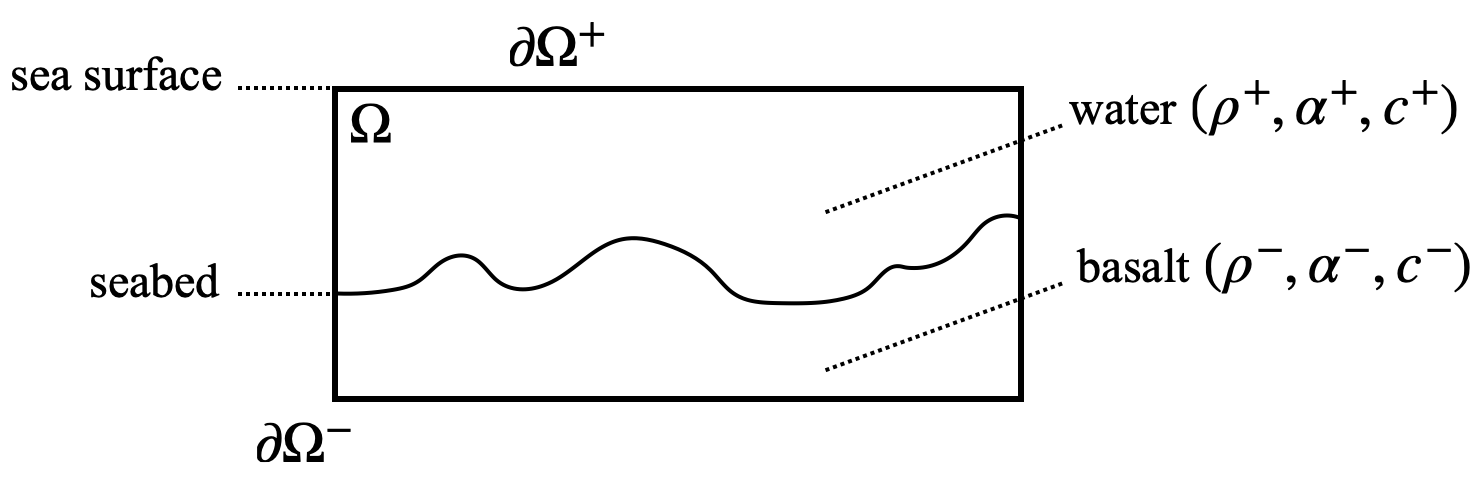}
    \caption{Information on the computational domain for seabed tomography.}
    \label{fig:seabed}
\end{figure}

We define the discrete solution operator $\mathcal W^{\delta}$, observation operator $\mathcal O_i^{\delta}$, and forward operator $\mathcal G_i^{\delta}:=\mathcal O_i^{\delta} \circ \mathcal W$, $i=1,\dots,N_f$, according to the discrete wave equation \eqref{eq:fully-discrete}. The discrete measurement model now takes the form
\begin{equation} \label{eq:inverse-discrete}
    \by^{\text{obs}}_i= \mathcal G^{\delta}_i(h) + \varepsilon_i, \qquad i=1,\dots,N_f.
\end{equation}
Here, by abusing the notation, we refer to $\by^{\text{obs}}_i$ to be the discrete measurement. Examples of the measurement is presented in \Cref{fig:measurement} which corresponds to the seabed $h$ illustrated in \Cref{fig:snapshots}.

\begin{figure}
    \begin{minipage}[b]{0.5\linewidth}
    \centering
    \includegraphics[width=\textwidth]{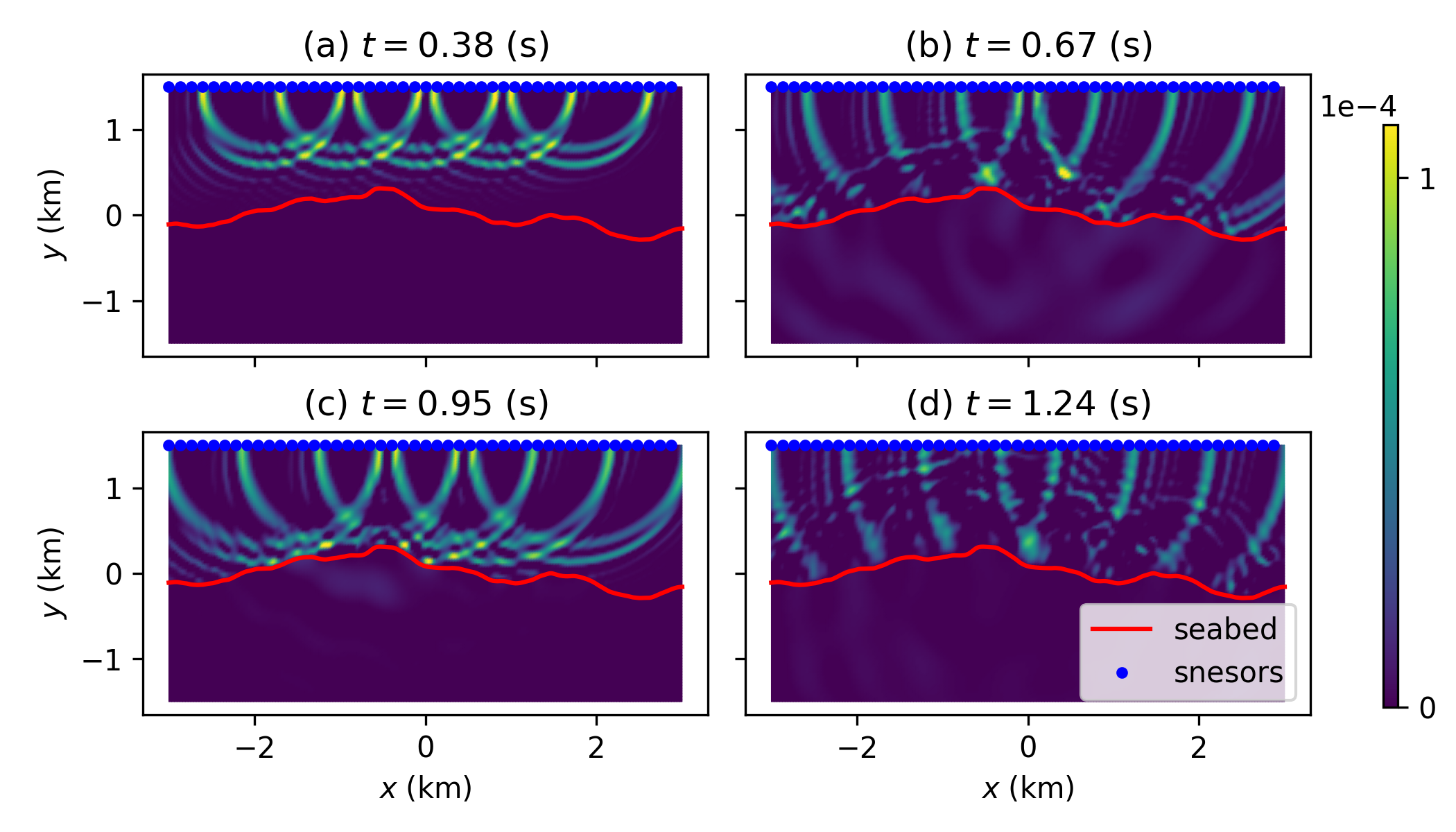}
    \caption{Time snapshots of the solution to the discrete problem \eqref{eq:fully-discrete}.}
    \label{fig:snapshots}
    \end{minipage}
    \begin{minipage}[b]{0.45\linewidth}
    \centering
    \includegraphics[width=\textwidth]{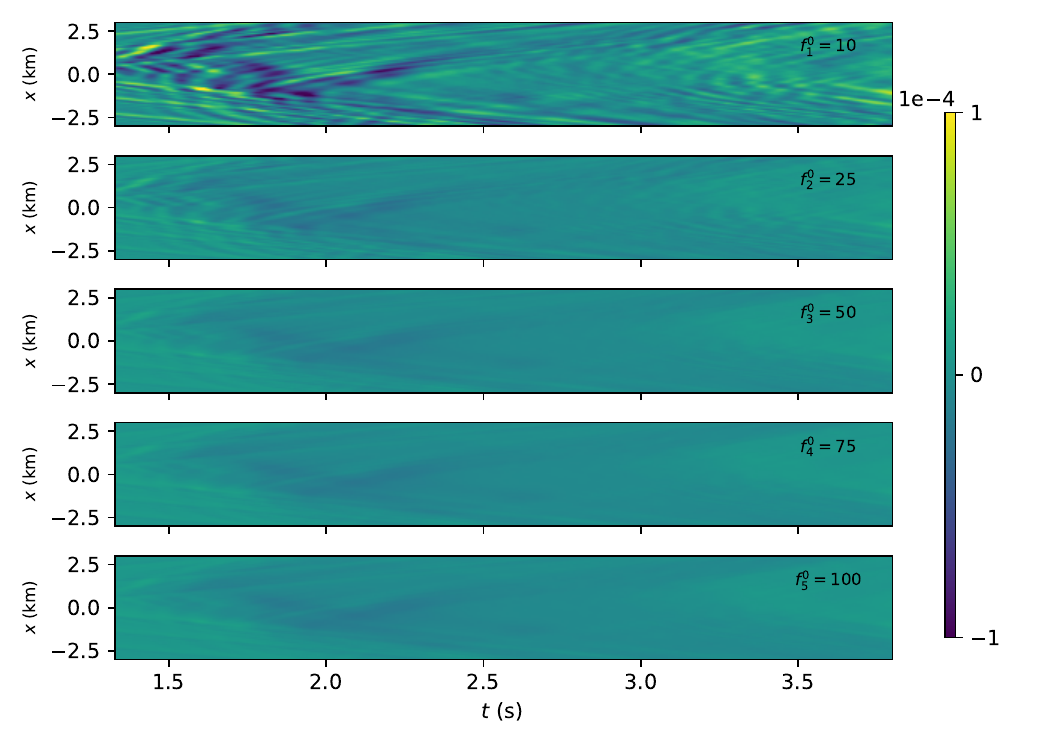}
    \caption{Information on the computational domain for seabed tomography.}
    \label{fig:measurement}
    \end{minipage}
\end{figure}

\section{Bayesian Formulation of Seabed Inference} \label{sec:bayes}
In this section we introduce the Bayesian approach to solving the inverse problem \eqref{eq:inverse-discrete}. 
In the Bayesian framework, we formulate quantities in the model, e.g., the unknowns and the measurement, in terms of random variables or fields. We encode our prior knowledge into this framework using infinite-dimensional probability distributions, i.e., probability measures. The solution to the Bayesian problem is then the conditional probability measure of the unknown (the seabed) given the measurement (displacements derived from the wave field at receivers).

The main components of a Bayesian inverse problem are the \emph{prior}, the probability distribution of the seabed in the absence of measurement, the \emph{likelihood} the probability distribution of the measurement for a given seabed, and the \emph{posterior}, the probability distribution of the seabed given measurement. In the following sections we will summarize these components. 

\subsection{Prior Modeling} \label{sec:prior}
In this section we discuss how to impose prior knowledge on the seabed based on its regularity. We first introduce the Hilbert scale spaces and discuss how it is associated with the regularity of functions. We then recall Gaussian distributions of Hilbert scale spaces.

Let $X=C([a,b])$ be the space of continuous functions and let $\{ e_j \}_{j=1}^{\infty}$ be an 
orthonormal basis, with respect to the $L^2([a,b])$ inner-product, for $X$. We define a norm for $f\in X$ as
\begin{equation} \label{eq:normHt}
    \| f \|^2_{H^\tau} :=\sum_{j=1}^{\infty} j^{2\tau}\left|  \langle f, e_j  \rangle \right|^2. 
\end{equation}
The Hilbert scale space $H^\tau$ is then defined by
\begin{equation}
    H^\tau([a,b]) := \{f\in X :  \| f \|_{H^\tau} < \infty \}.
\end{equation}
The following example, demonstrates the connection between the Hilbert scale parameter $s$ and the regularity, e.g., roughness vs. smoothness, of a function.

\paragraph{Example} Let $\{e_j\}_{j=1}^{\infty}$ represent the Fourier basis for $f \in X$, e.g., by periodic extension \cite{tolstov2012fourier} of the domain $[a,b]$ to the whole line. In this case $f$ has the Fourier expansion
\begin{equation}
    f = \sum_{j\in \mathbb Z} \hat f_j e_j,
\end{equation}
where $e_j = \exp({\bi 2 \pi \frac{j}{(b-a)}x})/z_j$, with  $\bi=\sqrt{-1}$, and $z_j\in \mathbb R^+$ is the normalization constant for $e_j$. Furthermore, $\hat f_j$ is the $L^2$ projection of $f$ onto $e_j$. We denote by $D^{\tau} f$  the fractional derivative of $f$ defined as
\begin{equation} \label{eq:fourier-series}
    (D^\tau f)(x) := \sum_{j\in \mathbb Z} (\bi 2 \pi \frac{ j}{b-a} )^\tau \hat f_j e_j.
\end{equation}
Note that this definition coincides with standard derivatives of $f$ for integer values of $\tau$. Therefore, $D^\tau f$ exists, if the \eqref{eq:fourier-series} is absolutely convergent, i.e., $\| f \|_{H^{\tau}}<\infty$. It is an easy exercise to verify that if $D^\tau f$ exists, then $D^s f$ exists for $s\leq \tau$. Therefore, $\|\cdot \|_{H^\tau}$ can be used to identify the regularity of $f$.

Now we introduce a Gaussian distribution for functions with regularity $\tau$. Let $(\Gamma, \mathcal A, \mathbb P)$ be a complete probability space \cite{ibragimov2012gaussian}, with $\Gamma$ a measurable space, $\mathcal A$ a $\sigma$-algebra defined on $\Gamma$, and $\mathbb P$ a probability measure defined on $\mathcal A$. We call $\eta$ to be an $H^{\tau}$-valued Gaussian random function, if $\langle \eta , \xi \rangle$ is a real-valued Gaussian random variable. The following theorem fully characterizes $\eta$ in terms of a mean function $m \in H^{\tau}$ and a trace-class, symmetric, and non-negative linear operator $\mathcal C:H^{\tau} \to H^{\tau}$.

\begin{theorem} \label{thm:guassian}
    \cite{ibragimov2012gaussian} Let $\eta$ be an $H^{\tau}$-valued Gaussian random function, then we can find $m \in H^{\tau}$ and trace-class, symmetric, and non-negative linear operator $\mathcal C:H^{\tau} \to H^{\tau}$, known as the covariance operator, such that
    \begin{equation}
        \begin{aligned}
            \langle m , \xi \rangle &= \mathbb E \langle \eta , \xi \rangle, \qquad &\forall \xi\in H^{\tau}, \\
            \langle \mathcal C \xi , \zeta \rangle &= \mathbb E \langle \eta - m , \xi \rangle \langle \eta - m , \zeta \rangle, \qquad &\forall \xi, \zeta \in H^{\tau},
        \end{aligned}
    \end{equation}
where $\mathbb E$ denotes expectation \cite{ibragimov2012gaussian}. We then define the push-forward probability measure $\mathcal N(m, \mathcal C) := \mathbb P \circ \eta^{-1}$. We denote $\eta \sim \mathcal N(m, \mathcal C)$ to refer that $\eta$ is a Gaussian random function define on the probability space $(H^{\tau}, \mathcal B (H^{\tau}), N(m, \mathcal C))$, where $\mathcal B (H^{\tau})$ is the Borel $\sigma$-algebra \cite{dudley2018real}. 
\end{theorem}
The following theorem recalls the \emph{Karhunen-Lo\'eve} (KL) expansion, which enables us to express $\eta$ in terms of the spectral decomposition of the covariance operator $\mathcal C$.

\begin{theorem} \label{thm:kl}
    \cite{ibragimov2012gaussian} Let $m$ and $\mathcal C$ be the mean and a covariance operator as in \Cref{thm:guassian}. Furthermore, let $\{e_j\}_{j=1}^{\infty}$ be the eigenfunctions and $\{ \lambda_j \}_{j=1}^{\infty}$ be the corresponding eigenvalues, sorted in decreasing order of the eigenvalues. Then, $\eta\sim \mathcal N(m,\mathcal C)$ if and only if $\eta$ has the infinite expansion
    \begin{equation} \label{eq:kl}
        \eta = m + \sum_{j=1}^{\infty} \sqrt{\lambda_j} \beta_j e_j,
    \end{equation}
    where $\beta_j\sim \mathcal N(0,1)$, $j \geq 1$, are independent standard normal random variables. We interpret the infinite summation as $\mathbb E \| \eta \|^2_{H^{\tau}} < \infty$.
\end{theorem}
In the seabed tomography problem, we reformulate the seabed, $h$ in \eqref{eq:wave-eq}, as a random function $\eta$ and impose a Gaussian prior distribution of the form $\eta \sim \mathcal N(m, \mathcal C)$. We choose $m$ and $\mathcal C$ such that the realizations of $\eta$ showcase the desired properties, e.g., the desired regularity. The KL-expansion \eqref{eq:kl} gives us the opportunity to construct $\mathcal C$ from its spectral decomposition, i.e., define the covariance $\mathcal C$ to be
\begin{equation} \label{eq:covariance-spectral}
    \mathcal C := \sum_{j=1}^{\infty} \lambda_j \langle e_j, \cdot \rangle e_j.
\end{equation}
This is a symmetric and linear operator. Choosing the sequence $\{\lambda_i\}_{i=1}^{\infty}$ such that $\lambda_i>0$, for all $i$, and $\sum_{i=1}^{\infty} \lambda_i <\infty$, ensures that $\mathcal C$ in \eqref{eq:covariance-spectral} is also trace-class and positive-definite. The following theorem characterizes the regularity of $\eta$ in terms of the decay rate of the eigenvalues $\lambda_j$.

\begin{theorem} \label{thm:smoothness}
    \cite{Dashti2017,dunlop2017hierarchical} Let there be constants $s,C^+,C^->0$ such that $\lambda_j \leq C^+ j^{-s }$, for integers $j\geq 1$, and define $\mathcal C_s$ to be the covariance operator following \eqref{eq:covariance-spectral}. Then $\mathbb E \| \eta - m \|_{H^{\tau}} <\infty$, with $\eta \sim \mathcal N(m, \mathcal C_s)$ and for all $\tau < s-1/2$.
\end{theorem}
\begin{pf}
    The complete proof can be found in \cite{Dashti2017}, however, here we provide a sketch of the proof. It follows
    \begin{equation} \label{eq:proof}
        \mathbb E \| \eta -m \|_{\Ht}^2 = \sum_{j=1}^{\infty} j^{2\tau} \lambda_j \| e_j \|^2_{\Ht} \mathbb E\beta_j^2  \leq \sum_{j=1}^{\infty} j^{\tau} C^+ j^{-s} = C^+ \sum_{j=1}^{\infty} j^{2(\tau - s)}.
    \end{equation}
    Here, we used the KL-expansion of $\eta$ and the fact that $\mathbb E \beta_i^2 = \text{Var}(\beta_i) = 1$. The summation on the right-hand-side of \eqref{eq:proof} is finite, if and only if $\tau < s-\frac{1}{2}$.
\end{pf}
\Cref{thm:smoothness} indicates that samples of $\eta\sim \mathcal N(m,\mathcal C_s)$ will belong to $H^s$ almost surely, for all $\tau < s-1$.

To avoid computational artifacts, we require the seabed to meet the boundaries of $\Omega$ with an orthogonal angle. To this end we choose the cosine basis, i.e., $e_j = \text{Im}\left( \exp(\bi\pi \frac{j}{b-a}x) \right)/\tilde{z}_j$, with $\text{Im}$ representing the imaginary part and $\tilde z_j$ being the normalization constant, to expand $\eta$. This is known as the cosine-basis \cite{briggs1995dft}, therefore, the fast cosine transform \cite{briggs1995dft} enables us to efficiently assemble the KL-expansion \eqref{eq:kl}. We refer the reader to \cite{maboudi2024} for a detailed numerical method. Furthermore, we take the eigenvalues $\{\lambda_j\}_{j=1}^{\infty}$ for $\mathcal C_{s}$ to take the form
\begin{equation} \label{eq:eigenvalues}
    \lambda_j = \left( (\frac{1}{\ell})^2 + j^2 \right)^{2s}, \qquad j\geq 1,
\end{equation}
for some $\ell > 0$. This choice of eigenvalues is inspired by the Whittle-Mat\'ern covariance \cite{roininen2014whittle} where $\ell$ is the length scale, controlling the radius of correlation between points on $\eta$. A sample of $\eta\sim \mathcal N(0, \mathcal C_s)$ is presented in \Cref{fig:prior}. In this sample $e_j$, $\beta_j$ and $m$, in the KL-expansion \eqref{eq:kl} is fixed, while the regularity parameter $s$, in \eqref{eq:eigenvalues} is varied. We see as $s$ increases, the realization of the seabed becomes smoother.

In the seabed tomography problem, a good estimate for $s$ is often challenging. We propose to infer the regularity parameter $s$ alongside the seabed $h$. Let $S$ be a real-valued random variable representing the regularity of the seabed with the prior probability measure $\pi_S$. For example, we can choose $\pi_S$ to represent the Lebesgue probability measure on the half-line $Y = \{1/2<s\}$. Recall that the lower bound in this interval is to ensure that the resulting Gaussian prior through the KL expansion above yields a valid Gaussian probability measure \cite{dunlop2017hierarchical}. Furthermore, we assume that the joint probability measure of $\eta$ and $S$, i.e., the prior measure, is given by $\pi^{\text{prior}}_{H,S} = \mathcal N(0, \mathcal C_s) \pi_S$, i.e., the probability of an open set $U\otimes V\subset X\times Y$ is
\begin{equation} \label{eq:prior}
    \pi^{\text{prior}}_{H,S}(U\otimes V) = \int_V \left( \int_U\ N(d\eta|0, \mathcal C_s) \right)  \ \pi_S(ds).
\end{equation}

\begin{figure}
    \centering
    \includegraphics[width=0.75\textwidth]{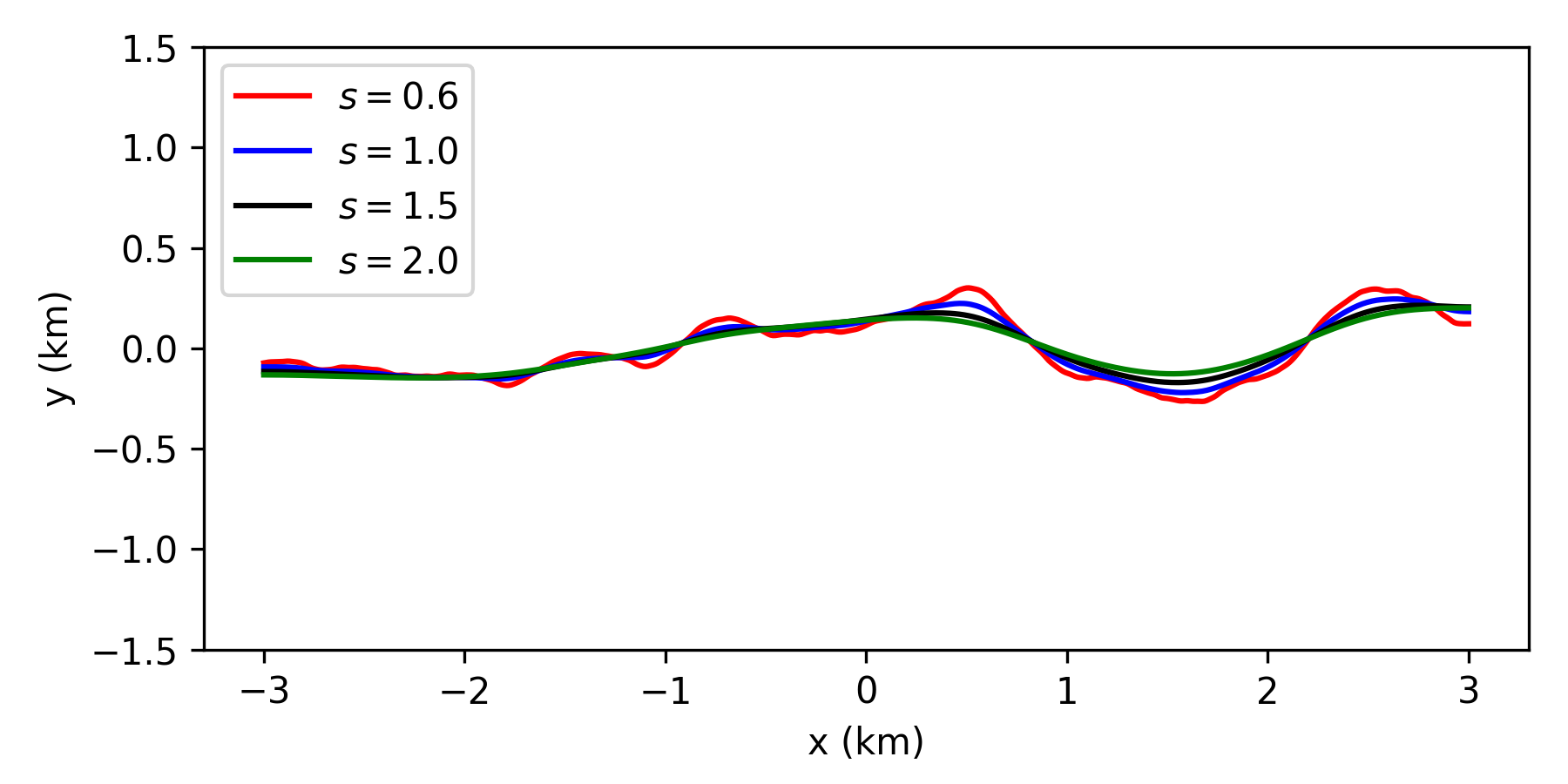}
    \caption{prior samples of $\eta$ for different regularity parameter $s$.}
    \label{fig:prior}
\end{figure}

\subsection{Likelihood} \label{sec:likelihood}
In this section, we reformulate the inverse problem introduced in \eqref{eq:inverse-discrete} in terms of random fields and variables. We let random vectors $\mathbf Y_i^{\text{obs}}, \eta$ and $\mathbf U$, represent the variables $\mathbf y_i^{\text{obs}}, h$ and $\mathbf u$, respectively. We also let $\mathbf y^{\text{obs}}$ and $\mathbf Y^{\text{obs}}$ be the concatenation of the vectors $\mathbf y_i^{\text{obs}}$ and $\mathbf Y_i^{\text{obs}}$, for $i=1,\dots,N_f$, respectively. The stochastic forward operator for the seabed tomography now takes the form
\begin{equation} \label{eq:inverse-stochastic}
    \mathbf Y^{\text{obs}}_i= \mathcal G^{\delta}_i(\eta, S ) + \varepsilon_i, \qquad i=1,\dots,N_f.
\end{equation}
Here, $S$ represents a real-valued random variable which represents the regularity of $\eta$, in the sense that $(\eta, S) \sim \pi_{H,S}^{\text{prior}} $. We assume that the observation random variables $\mathbf Y^{\text{obs}}_i$, $i=1,\dots,N_f$, are independent. Therefore, the likelihood function, i.e., the conditional distribution of the measurement for a given seabed configuration $\mathbf Y^{\text{obs}}| \mathbf H$, takes the product form
\begin{equation} \label{eq:joint-likelihood}
    \pi_{\mathbf Y^{\text{obs}}| H,S}(\by^{\text{obs}}) = \pi_{\mathbf Y^{\text{obs}}_1| H,S}(\by^{\text{obs}}_1) \dots \pi_{\mathbf Y^{\text{obs}}_{N_f}| H,S}(\by^{\text{obs}}_{N_f}).
\end{equation}
Furthermore, the relation \eqref{eq:inverse-stochastic} suggests the $i$th measurement follows the distribution 
\begin{equation}
    \left( \mathbf Y^{\text{obs}}_i - \mathcal G^{\delta}_i(H,S) \right) \sim \mathcal N(\mathbf 0, \sigma_i^{\text{noise}}I ).
\end{equation}
We can write the negative-log-likelihood (NLL) of \eqref{eq:joint-likelihood} in terms of the superposition
\begin{equation} \label{eq:negative-log-likelihood}
    \Phi( \mathbf y^{\text{obs}} ; \eta,s ) = \sum_{i=1}^{N_f} \Phi_i( \mathbf y^{\text{obs}}_i ; \eta,s ),
\end{equation}
where $\Phi_i$ is the NLL of $i$th measurement given as

\begin{equation} \label{eq:NLL-one-freq}
    \Phi_i( \mathbf y^{\text{obs}}_i ; \eta,s ) = \frac{1}{2\sigma_{i}^{\text{noise}}} \left( \sum_{1 \leq l \leq N_t} \Phi_i^l(\mathbf y^{\text{obs}}_i ; \eta,s)  \right), \qquad i=1,\dots, N_f.
\end{equation}
and
\begin{equation}
    \Phi_i^l(\mathbf y^{\text{obs}}_i ; \eta,s) =  \| [\mathbf y^{\text{obs}}_i]_l - [\mathcal G^{\delta}_i(\eta,s)]_l \|^2_{L^2(\partial \Omega^+)}, \qquad i=1,\dots, N_f,~ l=1,\dots,N_t.
\end{equation}
Note that $\|\cdot \|_{\partial \Omega^+}$ is norm of a function restricted to its boundary. We refer the reader to the \Cref{sec:appendix} for further details on this.

We remark that $\|\cdot \|_{\partial \Omega^+}$ in \eqref{eq:NLL-one-freq} implies a function-valued noise, whereas the noise model \eqref{eq:inverse-goal} is a point-wise noise model. In this work, we avoid a conflict between the two notations by integrating the noise over boundary cells when defining the FEM weak formulation of the measurement, i.e., noise is defined as
\begin{equation}
    \varepsilon := \sum^{N_{\text{FEM}}^{\text{bnd}}}_{i=1} [\varepsilon_j]_i \phi_{i}^{\text{bnd} },
\end{equation}
where $\phi_{i^{bnd}}$ is a boundary element, $N_{\text{FEM}}^{\text{bnd}}$ is the number of boundary elements on $\partial \Omega^+$ which contain a sensor, and $[\varepsilon_j]_i$ is the $i$th component of $j$th noise defined in \eqref{eq:inverse-goal}.

The following proposition we will show that the likelihood defined in \eqref{eq:negative-log-likelihood} fits in the wellposedness framework in \cite{stuart2010inverse}. In the rest of this article we may drop the superscript ``obs'' from the measurement vector and measurement random variables for simplification of expressions.
\begin{lemma} \label{thm:conditions}
    Let $(\Gamma, \mathcal A, \mathbb P)$ be the probability space associated with the priors introduced in \Cref{sec:prior}, and $\Phi^l_i( \mathbf y^{\text{obs}}_i ; \eta,s )$, for $i=1,\dots,N_t$, be the NLL of the $i$th frequency and $l$th time step. Furthermore, suppose that $\rho, \alpha, c  \in L^\infty(\Omega)$, with $\alpha_{\rm max} > \alpha > \alpha_{\rm min} > 0$, $\rho_{\rm max} > \rho > \rho_{\rm min} > 0$, $c > c_{\rm min} >0$. Then,
    \begin{enumerate}[label=(\roman*).]
        \item for any bounded measurement with $\|\mathbf y_i\|_{L^2(\partial \Omega^+)}\leq  r$, and $r>0$, we can find $C,\tilde C>0$ independent of $u$ such that
        \begin{equation} \label{eq:bounded1}
            0 \leq \Phi^l_i( \mathbf y^{\text{obs}}_i ; h(\eta,s) ) \leq C, ~ \text{and} ~
            0 \leq \Phi( \mathbf y^{\text{obs}}_i ; h(\eta,s) ) \leq \tilde C.
        \end{equation}
        \item For any fixed measurement $\mathbf y_i$ and bounded with respect to $\| \cdot \|_{L^2(\partial \Omega^+)}$, functions $\Phi^l_i( \mathbf y^{\text{obs}}_i ; h )$ and $\Phi( \mathbf y^{\text{obs}}_i ; h )$ are $\pi^{\text{prior}_{H,S}}$-a.s. continuous.
        \item For two measurements $\mathbf y_i^1$ and $\mathbf y_i^2$ with $\max \{ \| [\mathbf y_i^1]_l \|_{L^2} , \|[\mathbf y_i^2]_l\|_{L^2} \}_{l=1}^{N_t}<r$, with $r>0$, We can find $C>0$ depending on $r,\alpha_{\max},\rho_{\max}$, and $c_{\text{max}}$, such that 
        \begin{equation}
            | \Phi^l_i( \mathbf y^1_i ; h(\eta,s) ) - \Phi^l_i( \mathbf y^2_i ; h(\eta,s) ) |  \leq C \| [\mathbf y_i^1]_l - [\mathbf y_i^2]_l \|_{L^2},
        \end{equation}
        and
        \begin{equation} \label{eq:bounded2}
            | \Phi( \mathbf y^1_i ; h(\eta,s) ) - \Phi( \mathbf y^2_i ; h(\eta,s) ) |  \leq \tilde C \sum_{i,l} \| [\mathbf y_i^1]_l - [\mathbf y_i^2]_l \|_{L^2},
        \end{equation}      
    \end{enumerate}
\end{lemma}

\begin{pf}

\begin{enumerate} [label=(\roman*).]
    \item The expression is clearly bounded from below by 0. We show in Theorem 1 in Appendix 1 that $u(\mathbf x, \cdot ; h) \in H^1(\Omega)$. Therefore, according to the trace theorem (Theorem 5.36 in \cite{adams2003sobolev}) $[\mathcal G_i(\eta, s)]_l = u(\mathbf x, l \Delta t; h)|_{\partial \Omega^+} \in L^2(\partial \Omega^+)$. Therefore, there is $C\geq0$, independent of $u$ such that
    \begin{equation} \label{eq:trace}
        \| u(\mathbf x, \cdot ; h) \|_{L^2(\partial \Omega^+)} \leq C(T, \rho_{\text{min}},\|f\|_{C^2[0,T]}, \|g\|_{L^2(\Omega)}).
    \end{equation}
    It yields
    \begin{equation}
        \begin{aligned}
        \Phi^l_i( \mathbf y_i ; \eta,s ) &= \| [\mathbf y_i]_l - [\mathcal G_i(\eta,s)]_l \|^2_{L^2} \\
        & \leq \| [\mathbf y_i]_l \|^2_{L^2} + \| u_i(\mathbf x, l\Delta t ; h(\eta,s)) \|^2_{L^2(\partial \Omega ^+)} \leq r + C.
        \end{aligned}
    \end{equation}
    Inequality \eqref{eq:bounded1} follows from the fact that $\Phi$ is a finite linear combination of $\Phi^l_i$, for $i=1,\dots,N_f$ and $l=1,\dots,Nt$.
    \item It is sufficient to show the continuity result for $\Phi_i^l(\mathbf y_i^{\text{obs}};h)$, for a fixed $i$ and $l$, since $\Phi(\mathbf y_i^{\text{obs}};h)$ is a finite and linear combination of such log-likelihood functions.

    Recall that from \Cref{thm:A3} that we know that the map $h \mapsto u(h)|_{\partial \Omega ^+} $ is Lipschitz continuous in the sense that$\| u(h_1) - u(h_2) \|_{L^2(\partial \Omega^+)} \leq C \| h_1 - h_2 \|_{L^2}$. Furthermore, the map $v \mapsto \| v - \mathbf y_i^{\text{obs}} \|_{L^2}^2 $ is continuous (squared difference function). Therefore, the composite function, i.e., $\Phi_i^l(\mathbf y_i^{\text{obs}};h)$, is also continuous everywhere, therefore, it is $\pi^{\text{prior}}_{H,S}$-a.s. continuous.

    \item Fix $s$, and $\eta$ and consider $\mathbf y_i^1$ and $\mathbf y_i^2$ with $\max \{ \| \mathbf y_i^1 \|_{L^2} , \|\mathbf y_i^2\|_{L^2} \}<r$, and some $r>0$. It yields
    \begin{equation}
        \begin{aligned}
        |\Phi^l_i( \mathbf y^1_i ; \eta,s ) - \Phi^l_i( \mathbf y^2_i ; \eta,s ) | &\leq \frac{1}{2} |\langle [\mathbf y^1_i]_l + [\mathbf y^2_i]_l  - 2 [\mathcal G^\delta_i(\eta,s)]_l, [\mathbf y^1_i]_l - [\mathbf y^2_i]_l  \rangle_{L^2}| \\
        & \leq \frac{1}{2} ( \| [\mathbf y^1_i]_l \|_{L^2} + \| [\mathbf y^2_i]_l \|_{L^2} + 2\| [\mathcal G^\delta_i(\eta,s)]_l \|_{L^2} ) \\
        & \hspace{5cm} \|  [\mathbf y^1_i]_l - [\mathbf y^2_i]_l \|_{L^2} \\
        & \leq2 (r + \| [\mathcal G_i(\eta,s)]_l \|_{L^2} ) \| [\mathbf y^2_i]_l  - [\mathbf y^2_i]_l  \|_{L^2}.
        \end{aligned}
    \end{equation}
\end{enumerate}
Here, in the first inequality we reformulated the left-hand-side by adding and subtraction the term $\langle [\mathbf y^1_i]_l ,  [\mathbf y^2_i]_l \rangle_{L^2}$, and applying the Cauchy-Schwarz inequality in the second inequality and acknowledging that $[\mathcal G_i(\eta,s)]_l \in L^2(\partial \Omega^+)$ according to \eqref{eq:trace}. Similar to the argument in part $(i)$, \eqref{eq:bounded2} follows immediately.
\end{pf}

\subsection{Posterior}
In this section we use the Bayes' theorem to construct a posterior distribution, i.e., the goal-oriented distribution $\pi^{\text{post}}_{H,S|\mathbf y}$, for the seabed tomography.

\begin{theorem}
Let $(\Gamma, \mathcal A, \mathbb P)$ be the probability space defined in \Cref{sec:prior} associated with the prior distribution $\pi^{\text{prior}_{H,S}}$. Furthermore, suppose that $\Phi$ is the NLL defined in \Cref{sec:likelihood}. Then the posterior distribution $\pi^{\text{post}}_{H,S}$ (the conditional probability measure of the seabed given measurement) is absolutely continuous with respect to the prior distribution $\pi^{\text{prior}}_{H,S}$, i.e., $\pi^{\text{post}}_{H,S} \ll \pi^{\text{prior}}_{H,S}$, and is expressed as the Radon-Nikodym derivative
\begin{equation} \label{eq:Bayes}
    \frac{d \pi^{\text{post}_{H,S}}}{d \pi^{\text{prior}_{H,S}}} (\eta,s) = \frac{1}{Z} \exp(-\Phi(\mathbf y; \eta ,s)),
\end{equation}
with normalization constant
\begin{equation}
    Z = \int_{X}  \exp(-\Phi(\mathbf y; \eta,s) )  \pi^{\text{prior}}_{H,S}(d(\eta,s)).
\end{equation}
Furthermore, for two measurements $\mathbf y^1$ and $\mathbf y^2$ with $\max \{ \| [\mathbf y_i^1]_l \|_{L^2} , \|[\mathbf y_i^2]_l\|_{L^2} \}_{l=1}^{N_t}<r$, and $r>0$, there is $C>0$ independent of $u$, such that
\begin{equation}
    d_{\text{Hell}}( \pi^{\text{post}}_{H,S|\mathbf y^1}, \pi^{\text{post}}_{H,S|\mathbf y^2} ) \leq C \sum_{i,l} \| [\mathbf y^1_i]_l-[\mathbf y^2_i]_l \|_2.
\end{equation}
where, $d_{\text{Hell}}(\cdot,\cdot)$ is the Hellinger distance between probability measures \cite{le2000asymptotics}.
\end{theorem}
\begin{pf}
Bayes' theorem tells us that the relation between the posterior and the prior measure follows \eqref{eq:Bayes}. Therefore, to show that $\pi^{\text{post}}_{H,S} \ll \pi^{\text{prior}}_{H,S}$ we show that the right-hand-side of \eqref{eq:Bayes} is well-defined, i.e.  $\Phi$ is measurable, and $Z$ is finite and positive.

In \Cref{thm:conditions} we showed that $\Phi(\mathbf y; \cdot)$ and $\Phi(\cdot; \eta,s)$, are a.s. continuous, and locally Lipschitz, respectively. Therefore, the mapping $\Phi(\cdot;\cdot)$ is jointly continuous, $\pi^{\text{prior}}_{H,s}$-a.s., and therefore, it is $\pi^{\text{prior}}_{H,S}$-measurable.% (see Lemma 6.1 in \cite{iglesias2016bayesian}).

To show that $Z>0$ recall that $\Phi$ is bounded from above according to \Cref{thm:conditions}$(i)$. Therefore,
\begin{equation}
    \int \exp(-\Phi(\mathbf y;\eta,s)\pi^{\text{prior}}_{H,S}(d(\eta,s)))\geq \exp(-\tilde C) \int \pi^{\text{prior}}_{H,S} (d(\eta,s))) = \exp(-\tilde C) > 0.    
\end{equation}
Showing that $Z$ is bounded from above is yielded by the fact that $\exp(-\Phi(\mathbf y;\eta,s) )\leq 1$. Therefore, Bayes' theorem applies and $\pi^{\text{post}}_{H,S} \ll \pi^{\text{prior}}_{H,S}$. Hellinger well-posedness can be derived identically to Theore 2.2 in \cite{iglesias2016bayesian}.
\end{pf}

\section{Numerical Algorithm to Explore The Posterior}

In the previous sections we showed that the goal-oriented posterior for the seabed tomography problem is a proper probability measure, and that it is locally continuous with respect to the measurement. In this section we discuss a numerical method to explore the posterior, estimate the average seabed and quantify the uncertainties in the estimation.

Sampling strategies, e.g., Markov chain Monte Carlo (MCMC) \cite{le2000asymptotics} methods, are popular tools to explore intractable posterior distributions. These methods extract a set of approximate realizations of the posterior distribution which can estimate moments of the posterior. We propose a Metropolis-within-Gibbs (MWG) \cite{le2000asymptotics} sampling strategy for exploring the posterior for the seabed tomography problem.

The MWG method is an MCMC method which repeatedly draws samples from conditional distributions $\pi_{H|S,\mathbf y}$ and $\pi_{S|H,\mathbf y}$. The procedure is stopped when the desired number of samples is obtained. This approach is suitable for this problem due to the hierarchical structure of the prior (see \eqref{eq:prior}) \cite{le2000asymptotics}. In the following, we propose suitable strategies to draw samples from each of these conditional distributions.

\subsection{Sampling from \texorpdfstring{$\pi_{H|S,\mathbf y}$}{Lg} } \label{sec:fixed-s} We can symbolically break down this conditional distribution as $\pi_{H|S,\mathbf y}\propto \pi_{\mathbf y|H, S} \pi_{H|S}$. Recalling that $\pi_{H|S} = \mathcal N(0, \mathcal C_s)$, the distribution $\pi_{H|S,\mathbf y}$ is a posterior with a Gaussian prior. Efficient sampling methods, e.g., the preconditioned Crank-Nicolson (pCN), are developed in \cite{e6c3f9b1-115e-37a9-b4d4-010ab9364de8} for such posteriors. A brief description of the pCN is as follows. Suppose at the $i$th step of the Markov chain we obtain the sample $h_i$. We then propose the sample
\begin{equation} \label{eq:pCN}
    h^\star = (1 - \beta_h^2)^{1/2} h^i + \beta_h \xi,
\end{equation}
where $\xi$ is a realization from the prior distribution $\mathcal N(0, \mathcal C_s)$, and $\beta_h>0$ is the step size. We then accept this sample with probability
\begin{equation}
    a_h(h^i, h^\star) := \min \left\{ 1,\exp(\Phi(\mathbf y; h^i, s) -  \Phi(\mathbf y; h^\star, s) ) \right\}.
\end{equation}
The new sample in the Markov chain, $h^{i+1}$, is $h^\star$, if the sample is accepted, and is $h^i$ if the sample is rejected.

\subsection{Sampling from \texorpdfstring{$\pi_{S|H,\mathbf y}$}{Lg} }
We may break down this distribution as $\pi_{S|H,\mathbf y} \propto \pi_{\mathbf y|S, H} \pi_{S|H}$. However, understanding the distribution $\pi_{S|H}$ is more involved.

To elaborate, suppose we have a sample $h$ and we want to compare two regularity parameters $s^1$ and $s^2$ (e.g., in an acceptance/rejection step) for $h$. Therefore, we must compare the probabilities of $h$, or more precisely small open set $dh$ around $h$, according to distributions $\mathcal N(0, \mathcal C_{s^1})$ and $\mathcal N(0, \mathcal C_{s^2})$. However, according to \cite{feldman1958equivalence}, the probability measures $\mathcal N(0, \mathcal C_{s^1})$ and $\mathcal N(0, \mathcal C_{s^2})$ are mutually singular \cite{schervish2012theory} when $s^1\not = s^2$, i.e., the probability of $dh$ is zero for almost all $s\in Y$. Therefore, a hierarchical Bayesian selection of $s$ is not possible in this setting. See \cite{schervish2012theory,dunlop2017hierarchical} for more detail.

Instead, we take a frequentist approach and choose $s$ solely on its likelihood probability, i.e., we approximate $\pi_{S|H,\mathbf y} \approx\pi_{\mathbf y|S, H}$. We then draw samples for $s$ following a Metropolis-Hasting \cite{kaipio2006statistical} method.

Let $q:Y\times Y \to [0,1]$ be a proposal distribution (in this paper we choose a Gaussian i.i.d. proposal distribution with a small variance), associated with a step size $\beta_s>0$, for $s$ and let $s_i$ be the $i$th sample of the Markov chain. We propose $s^{\star}$ from the proposal distribution $q(s_i,s^{\star})$. We then accept this sample with probability
\begin{equation}
    a_s(s^i, s^\star) := \min \left\{ 1,\exp(\Phi(\mathbf y; h, s^i) -  \Phi(\mathbf y; h, s^{\star}) ) \right\}.
\end{equation}
The new sample in the Markov chain, $s^{i+1}$, is $s^\star$, if the sample is accepted, and is $s^i$ if the sample is rejected.

%\textcolor{blue}{Ana:\ What proposal distribution q do you choose?}

The problem of estimating the regularity of an in-homogeneous medium, from a noisy likelihood, has been considered in the context of inverse scattering problems \cite{karamehmedovic2013efficient,schroder2011modeling}. Furthermore, the authors in \cite{sym13091702} discuss the relation between the emission frequencies, $f_i$ in \eqref{eq:freqs}, and detecting smallest variations in an in-homogeneous medium.% \babak{Ana, would you like to add something?}. 

It is demonstrated in \cite{maboudi2024} that the sampling framework discussed here is robust to noisy and partial data and suitable for various ill-posed inverse problems. The performance of this method is further discussed in the result section down below. We summarize the complete sampling algorithm in \Cref{alg:1}.

We note the availability of gradient-assisted sampling strategies, including unadjusted and Metropolis-adjusted Langevin algorithms \cite{roberts1998optimal} as well as Hamiltonian Monte Carlo \cite{duane1987hybrid} and its variants, which are widely used for sampling complex posterior distributions. In the present work, however, the parameterization of the seabed introduces non-smoothness at the water–seabed interface, arising from the piecewise-smooth structure of the prior on density and elastic properties. This leads to discontinuous or poorly defined gradients of the forward operator with respect to the unknown topography, which undermines the assumptions required for stable and efficient gradient-based sampling.

While one could introduce smooth approximations of the prior near the seabed interface to enable differentiability, such regularization would directly affect the small-scale roughness that constitutes the primary quantity of interest in this study, potentially biasing the inferred topography. For this reason, we restrict attention to non-gradient-based sampling strategies. We nevertheless include a comparison of the proposed approach with a functional ensemble sampling method \cite{coullon2021ensemble}, as modern technique that does not require gradient information of the log-likelihood.

\begin{algorithm} 
\caption{Drawing samples from the posterior distribution $\pi_{H,S|\mathbf y}$} \label{alg:1}
\label{alg:Gibbs}
\begin{algorithmic}[1]
\STATE{Input: initial states $h^0$ and $s^0$, step sizes $\beta_h$ and $\beta_s$, and flag $phase \in \{ \text{warm-up}, ~\text{online} \}$}
\WHILE{ $i<N_{\text{sample}}$ }
    \STATE $h \leftarrow h^i$
    \WHILE{ $j<N^h_{\text{inner}}$ }
        \STATE Propose $h^\star$ according to \eqref{eq:pCN}.
        \STATE Draw $u \sim \mathcal U([0,1])$
        \IF{$a_h(h,h^\star) > u$}
            \STATE $h \leftarrow h^\star$
        \ENDIF
        \IF{$phase$ is warm-up}
            \STATE tune $\beta_h$
        \ENDIF
    \ENDWHILE
    \STATE $h^{i+1} \leftarrow h$
    \STATE $s \leftarrow s^i$
    \WHILE{ $j<N^s_{\text{inner}}$ }
        \STATE Propose $s^\star$ according to the proposal distribution $q$.
        \STATE Draw $u \sim \mathcal U([0,1])$
        \IF{$a_s(s,s^\star) > u$}
            \STATE $s \leftarrow s^\star$
        \ENDIF
        \IF{$phase$ is warm-up}
            \STATE tune $\beta_s$
        \ENDIF
    \ENDWHILE
    \STATE $s^{i+1} \leftarrow s$
\ENDWHILE
\RETURN Set of samples $\{ (h^i,s^i) \}_{i=1}^{N_{\text{sample}}}$
\end{algorithmic}
\end{algorithm}

\section{Numerical Experiments and Results} \label{sec:results}
In this section we demonstrate the performance of the proposed method in a series of numerical experiments. We first discuss the numerical set up and the discretization details. The test cases in this section comprise a case where the regularity of the seabed is known but its location is unknown, and a case where both the regularity and the shape is unknown. Finally we perform the method on a seabed sample that is not a sample of the prior. %\textcolor{blue}{Ana:\ what do you mean by this?}

\subsection{Numerical Details}

In this section we summarize the numerical details of the forward evaluation, and the sampling procedure.

We discretize the domain $\Omega$ with a regular structured mesh of size $188\times 95$, resulting in cell sizes of size $h_x=h_y\approx 0.032$. We choose first-order Lagrangian elements to discretize functions on this mesh with approximately $18000$ degrees of freedom. To assemble \eqref{eq:fully-discrete} we choose a time-step $\Delta t = 0.0019$. We then solve \eqref{eq:fully-discrete} up until maximum time $t_{\text{max}}=3.95~\text{s}$ to obtain an approximate solution $\mathbf u^n$, at time steps $n=1,\dots,N_{t}$. Model parameters for elasticity and density in all simulations are $\rho_0=1$, $\rho^-=3$ $\lambda_0=1.5$, $\alpha_{\text{bottom}}: = \lambda^- + 2\mu^- =6.4$.

To obtain a noise-free observation vector, we record the values of $\mathbf u$ at the boundary $\partial \Omega^+$ at location of the sensors. We place 186 equidistant sensors across $\partial \Omega^+$. The computational mesh is designed to always contain nodes at the location of the sensors, even when the mesh is refined. This will avoid extra computational inaccuracy associated with interpolating solution values onto sensor locations. Furthermore, we assume that the measurement frequency is $0.0019$s. We present snapshots of the solution for one realization of the seabed and the corresponding measurement time series in \Cref{fig:snapshots,fig:measurement}, respectively.

We accelerate calculations by parallelizing forward evaluations. We first exploit the trivially parallel nature of the measurement models \eqref{eq:inverse-discrete} for different frequencies and perform forward evaluations in parallel. Furthermore, we perform parallel conjugate gradient (CG) iterations for solving \eqref{eq:fully-discrete}. We use the built-in CG routine in the FEniCS python package over 60 computational nodes.

The prior for the seabed tomography is constructed according to \Cref{sec:prior}, using the KL expansion \eqref{eq:kl} with cosine-basis and and eigenvalues in \eqref{eq:eigenvalues}. We truncate the KL expansion after $N_{\text{KL}}=256$ terms and discretize the basis functions $e_j$ on a uniform mesh of size 512. We construct the NLL and the goal-oriented posterior according to \eqref{eq:negative-log-likelihood} and \eqref{eq:Bayes}, respectively. 

To obtain the data, for all simulations in this paper, we perform forward estimations on a fine mesh of size $377\times189$, with approximately 70000 degrees of freedom. To add noise to the noise-free data we follow the forward model \eqref{eq:inverse-discrete} with noise standard deviation corresponding to
\begin{equation}
    \sigma_i^{\text{noise}} := 0.01 \times \text{max} \{ \| \mathbf y_i(t_1) \|_{2}, \dots , \| \mathbf y_i(t_{N_t}) \|_{2} \},
\end{equation}
where $\mathbf y_i(t_j)$, for $i=1,\dots,N_f$ and $j = 1,\dots,N_{t}$, are noise free signals. The inference methods are then performed with this data and the observation operators are evaluated with the coarse mesh.

\subsection{Fixed regularity \texorpdfstring{$s$}{Lg}} \label{sec:fixed-regularity}
In the first numerical experiment, we fix the regularity parameter to $s=0.75$. To create the data, we choose   random expansion coefficients in $\eqref{eq:kl}$, and perform a forward simulation on the fine grid. The true expansion coefficients will remain unknown to the inference method.

The inverse problem with a fixed regularity $s$ is a standard, nonlinear Bayesian inverse problem with a Gaussian prior. We apply the pCN MCMC method to draw samples from this posterior. To construct a sampling method for this problem see \Cref{sec:fixed-s} and we refer the reader to \cite{e6c3f9b1-115e-37a9-b4d4-010ab9364de8} for detailed derivations. As a comparison, we also investigate the performance of the functional ensemble sampler (FES) \cite{coullon2021ensemble}.

Affine invariant ensemble samplers \cite{goodman2010ensemble} are well suited to posterior distributions exhibiting strong anisotropy and poor scaling, as they are affine invariant and therefore insensitive to linear reparameterizations such as scaling and rotation of the parameter space. However, classical ensemble samplers are designed for finite-dimensional settings and are not directly applicable to infinite-dimensional Bayesian inverse problems.
To address this limitation, FES is employed within a hybrid Gibbs sampling framework. The unknown field is represented via a KL expansion, which is partitioned into a finite set of low-frequency coefficients (typically the first 5–10 modes and in this experiment we choose 10 first KL modes) and a complementary high-frequency component. The low-frequency coefficients are updated using an affine-invariant ensemble move, while the remaining coefficients are updated using a pCN scheme. These updates are alternated within a Gibbs sampling strategy.

Within the ensemble update, proposed states are generated through stretch moves followed by a Metropolis–Hastings type acceptance step. The stretch parameter controls the size of ensemble proposals (referred to as walkers) and is typically chosen in the range $[1.5,2.5]$ to ensure efficient exploration. When the stretch parameter approaches unity (e.g., $\ll 1.2$), ensemble proposals become highly localized, leading to strongly correlated samples and a loss of the affine-invariant advantage. In this regime, the sampler effectively degenerates to a slow random-walk behavior, which undermines reliable posterior estimation. In this experiment, we consider an ensemble of 40 walkers for the 10 first modes of the KL expansion following the recommendation in \cite{emcee}. We report that our experiments shows similar behaviour for fewer number of KL modes chosen for the FES method.

We divide the sampling procedure in 2 phases; the warm-up (burn-in) phase and the online phase. To warmup the pCN sampler, we apply the method discussed in \cite{sherlock2009optimal}, over 10K samples, to tune the step-size of the pCN method to obtain a tuned step-size $\beta_h \approx 0.007$. In the sampling phase, we fix the step size of the pCN method and draw 20K samples from the posterior. All numerical experiments in this manuscript were conducted on a workstation equipped with an AMD Ryzen Threadripper PRO 7975WX 32-core CPU. The total memory consumption for a complete sampling run was approximately 10.5~GB, with the majority attributable to the finite element simulations. The warm-up phase required approximately 1.5 days, followed by an additional 3 days for the sampling phase. We report that the pCN step size obtained after the warm-up phase is $\beta_h=0.007$. For the FES method, warm-up is performed following the procedure described in \cite{emcee}. During this phase, the adaptive tuning of the stretch parameter consistently drives it to a hard lower bound of 1.2, after which it is fixed for the remainder of the sampling phase. In both cases, samples generated during the warm-up stage are discarded prior to the evaluation of posterior statistics.

We note that the contraction of the stretch parameter in FES during the warm-up phase raises concerns regarding the effectiveness of the method for posterior exploration. The resulting small stretch parameter suggests that the anisotropy present in the posterior cannot be adequately mitigated by affine transformations alone, indicating a posterior geometry that is strongly influenced by nonlinear and nonsmooth features arising in the seabed tomography problem. We interpret this behavior as further evidence of the intrinsic complexity of the posterior, and as an indication that approaches relying on smooth parameterizations and gradient-assisted sampling may also encounter difficulties during warm-up and tuning in this setting. A systematic investigation of sampling strategies better suited to the nonlinear and nonsmooth posterior structure arising in this problem is left for future direction of research.

We present the trace plot of the first 5 coefficients $\beta_j$, in the KL expansion of $\eta$, for the pCN method in \Cref{fig:pCN-trace}. We notice that after the warm-up period, there is only a mild correlation between the successive samples. This means we can draw reliable estimates of the posterior from the samples. However, we see that some coefficients, e.g. $\beta_5$, show a stronger correlation between successive samples. This indicates that exploring the posterior with an MCMC method is a challenging task. We report that the performance of the pCN is similar for the rest of the coefficients. The trace plot for the FES method is shown in \Cref{fig:FES-trace}. Although the samples may appear to exhibit improved mixing, this apparent behavior arises from the presence of an ensemble of walkers rather than from efficient exploration of the posterior. The step-like behavior observed in the trace plot is consistent with poor interpolation between ensemble members, which is a known artifact associated with a small stretch parameter.

We show the posterior mean for the KL expansion coefficients, $\beta_j$, and the associated 95\% highest posterior density interval in \Cref{fig:pCN-HPD,fig:FES-HPD}, corresponding to the pCN and FES methods, respectively. We notice that pCN estimates few first coefficients are predicted with high accuracy and high certainty. Furthermore, the level of uncertainty (length of HPD) is comparable for the rest of the coefficients. Therefore, the method can estimate the all coefficient (associated with the frequencies of the seabed) with similar level of certainty. We emphasize that for higher frequencies, there is a discrepancy between the posterior mean and the true parameters. This is expected and reported in \cite{afkham2024bayesian}. Nonetheless, we can still recover the  density, $\rho$, and the elastic coefficients 
$\alpha$ that correspond to the true density and elastic coefficients. The predictions obtained using FES are less accurate than those produced by pCN, and the associated uncertainty estimates are substantially larger. In light of the discussion above, this elevated level of uncertainty is likely attributable to ineffective interpolation among ensemble walkers, rather than to genuine posterior variability.

\begin{figure}
    \centering
    % --- First row: pCN ---
    \subfloat[trace plot for $\beta_j$ (pCN)]{
        \label{fig:pCN-trace}
        \includegraphics[width=0.45\textwidth]{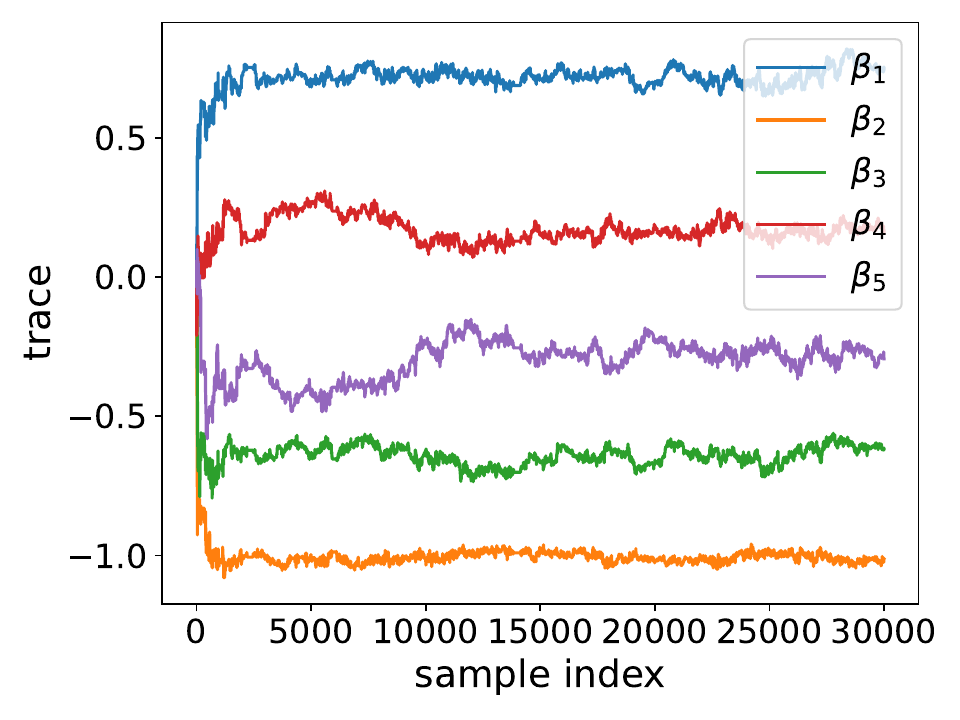}
    }
    \hfill
    \subfloat[estimation of the coefficients (pCN)]{
        \label{fig:pCN-HPD}
        \includegraphics[width=0.45\textwidth]{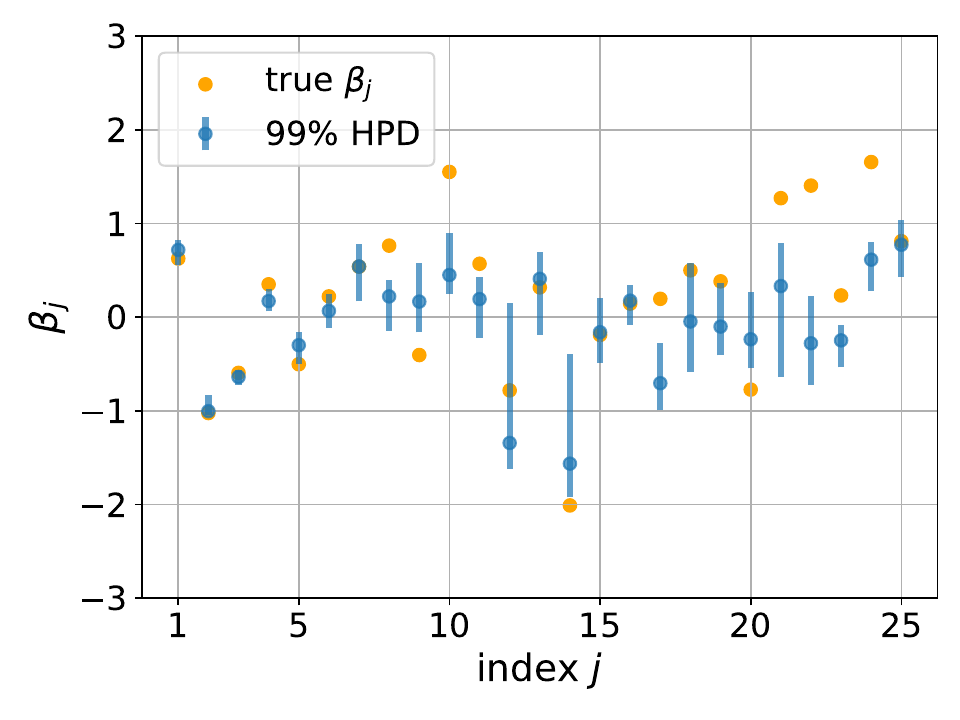}
    }\\[0.5em]
    % --- Second row: AIES ---
    \subfloat[trace plot for $\beta_j$ (FES)]{
        \label{fig:FES-trace}
        \includegraphics[width=0.45\textwidth]{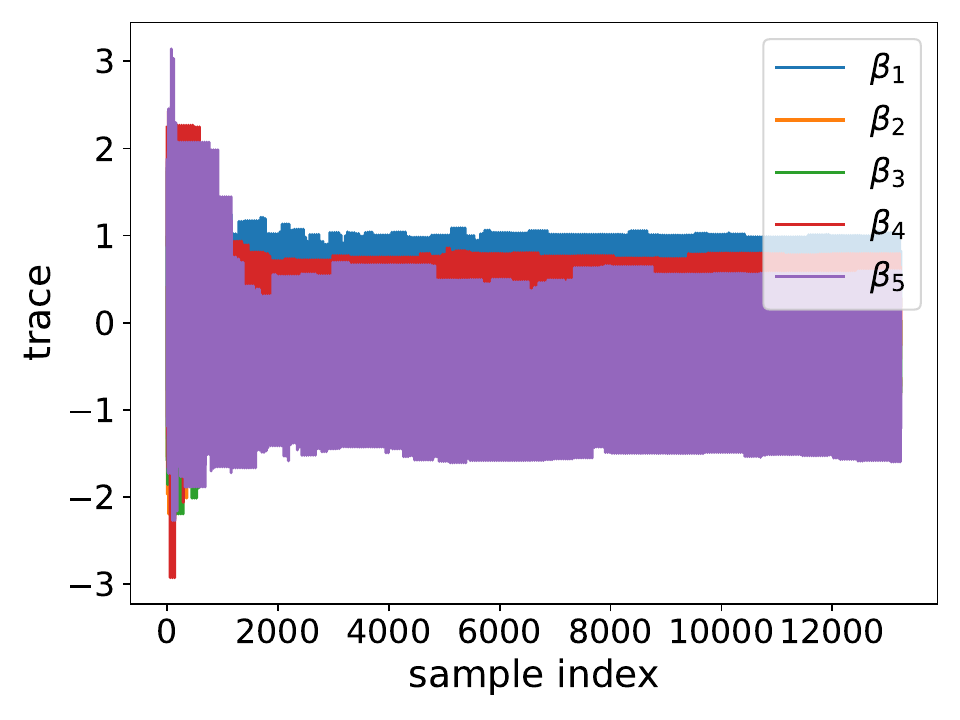}
    }
    \hfill
    \subfloat[estimation of the coefficients (FES)]{
        \label{fig:FES-HPD}
        \includegraphics[width=0.45\textwidth]{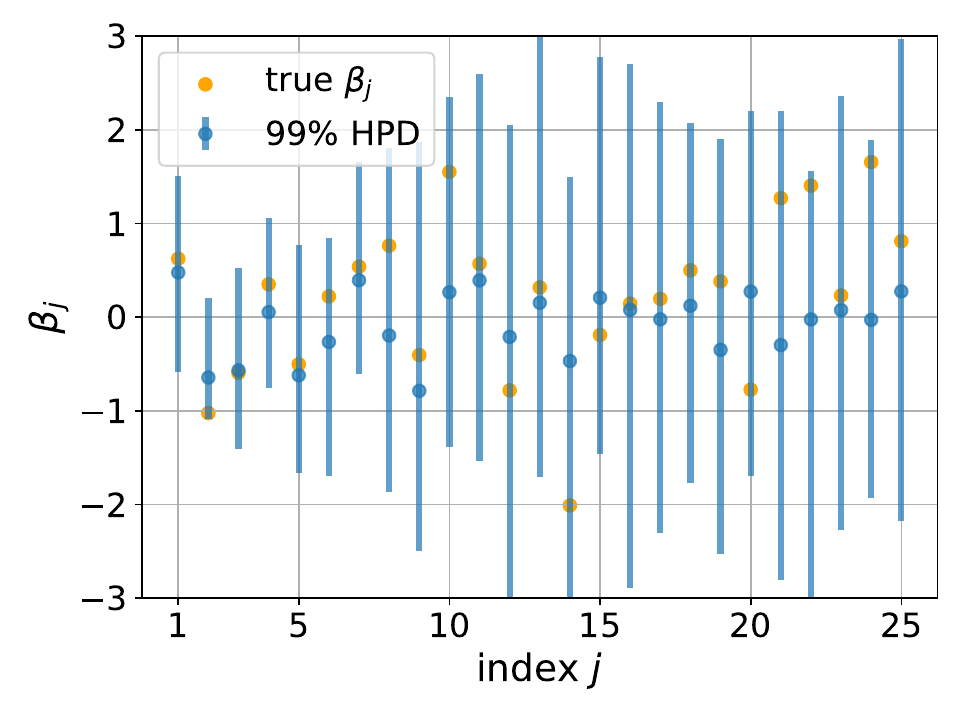}
    }
    \caption{Diagnostic plots for the pCN and r methods. Top row: pCN results — (a) trace plot of the first 5 coefficients $\beta_j$, (b) posterior mean and highest posterior density interval for each coefficient. Bottom row: FES results — (c) trace plot of the first 5 coefficients, (d) posterior mean and highest posterior density interval.}
    \label{fig:pcn-aies-diagnostics}
\end{figure}

The estimated seabed and its associated uncertainty is provided in \Cref{fig:pCN-mean-uq,fig:FES-mean-uq}. We estimate the mean seabed, by computing the mean coefficients $\bar{\beta}_j:= \mathbb E(\beta_j)$, $j=1,\dots,N_{\text{KL}}$, and then assembling the sum $\bar h := \sum_{j<N_{\text{KL}}} \sqrt{\lambda_j} \bar{\beta}_j e_j$. We present the level of uncertainty in the estimate in the form of a 99\% credibility band, which is obtain by interpolating point-wise 99\% credibility intervals (CIs).

We notice pCN provides an excellent match between the mean seabed and the true seabed, particularly in the domain interest $x\in[-2,2]$. Furthermore, the CI band expands outside the domain of interest, particularly where the true seabed diverges from the estimated seabed. The FES method captures the overall shape and qualitative behavior of the seabed, and the associated uncertainty band envelopes the true seabed location. However, the method fails to produce an accurate posterior mean estimate. We report that we repeated the FES experiment for 5 and 7 number of modes (instead of 10) and the same number of walkers and received similar behavior.

\begin{figure}
    \centering
    \subfloat[pCN]{
        \label{fig:pCN-mean-uq}
        \includegraphics[width=0.45\textwidth]{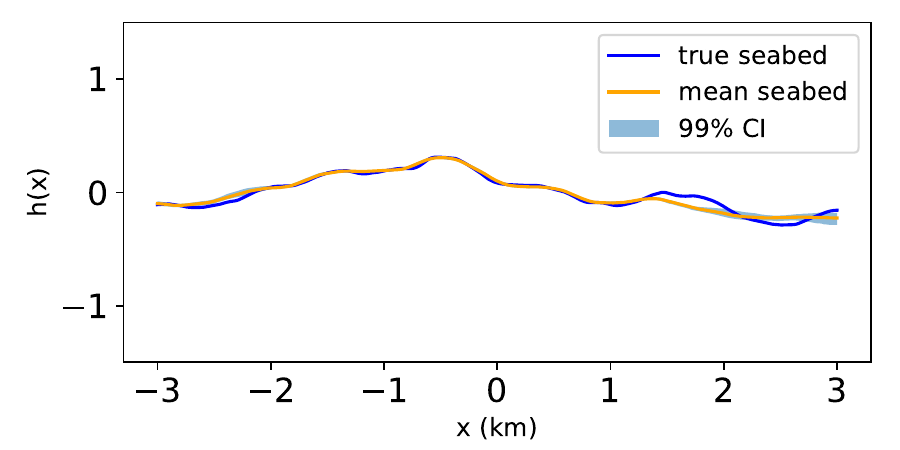}
    }
    \hfill
    \subfloat[FES]{
        \label{fig:FES-mean-uq}
        \includegraphics[width=0.45\textwidth]{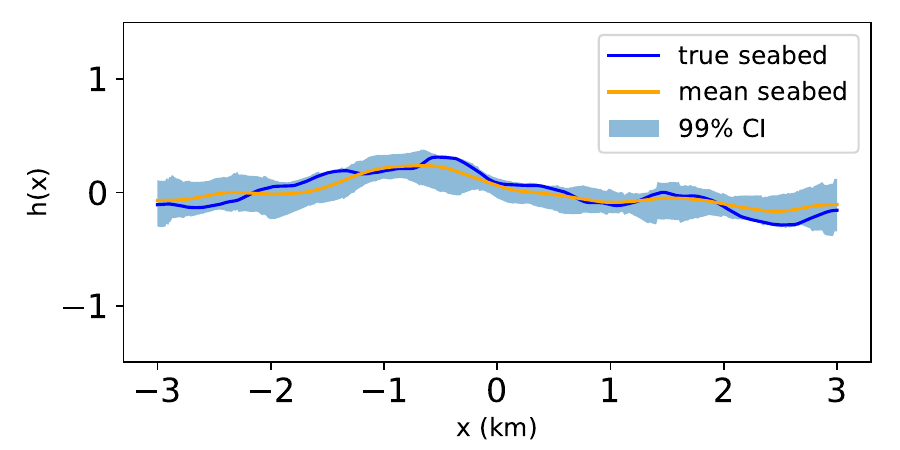}
    }
    \caption{Estimated seabed and quantified uncertainty with 99\% highest posterior density region. Left: pCN. Right: FES.}
    \label{fig:pcn-aies-mean-uq}
\end{figure}

To conclude this section, we remark that this example demonstrates that when the regularity parameter $s$ is known, we can estimate the seabed with relatively few pCN samples with high accuracy. The pCN method provides robust performance for this problem within an established framework for infinite-dimensional priors, such as the one considered here. In contrast, FES degrades in the presence of nonsmooth posterior geometry and exhibits poor practical mixing for this problem. Consequently, all remaining numerical experiments in this paper are carried out using the pCN method.

\subsection{Uncertain regularity \texorpdfstring{$s$}{Lg}} \label{sec:uncertain-regularity}

In this experiment we repeat the experiment from \Cref{sec:fixed-regularity} and generate data accordingly. However, we treat the regularity of the seabed as an unknown parameter of the problem.

To construct the goal-oriented posterior, we follow \Cref{sec:bayes}. We construct the joint prior as \eqref{eq:prior}. We choose a non-informative prior for the regularity and set $\pi_S = \mathcal U([0.5, 5])$, i.e., the uniform distribution on the interval $[0.5, 5]$. We note that the upper bound introduced here serves as a numerical convenience, and in practice the samples do not reach this bound. The primary challenge in estimating the regularity parameter lies in distinguishing rough seabeds, that is, cases with $s \ll 1.5$. In contrast, scenarios with $s>2$ are often visually indistinguishable.

We construct the likelihood function by forming the NLL as in \eqref{eq:negative-log-likelihood}. We then perform \Cref{alg:Gibbs} to draw samples from the posterior.

We divide the sampling procedure into 2 phases. In the warm-up phase, we tune the step sizes $\beta_h$ and $\beta_s$, in \Cref{alg:Gibbs}. We set $phase$ to be warm-up and choose the inner loop of size $N^h_{\text{inner}} = N^s_{\text{inner}} = 250$. In total, we choose $N_{\text{sample}} = 20$ warm-up samples. Together, these choices result in a total of 10k warm-up samples with tuned step-sizes $\beta_h\approx 0.003$ and $\beta_s \approx 0.017$. The warm-up method for this experiment follows from \cite{sherlock2009optimal}.

In the online phase, we fix the step sizes $\beta_h$ and $\beta_s$, in \Cref{alg:Gibbs} and we set $phase$ to be online. We then choose $N^h_{\text{inner}} = N^s_{\text{inner}} = 20$ and draw $N_{\text{sample}} = 1000$, resulting 40k total samples in the second phase. To estimate posterior statistics, we discards the samples obtained in the warm-up phase.

We show the trace plot of the first 5 expansion coefficients of $h$ in \Cref{fig:Gibbs-trace} after the warm-up phase. We observe that the traces show a similar pattern to the previous section. The HPD intervals of the coefficients is presented in \Cref{fig:Gibbs-HPD}. We notice that this plot is comparable with the case in \Cref{sec:fixed-regularity}, showing larger uncertainty for higher frequency coefficients. This indicates that the Metropolis-within-Gibbs approach in this paper is effective in reconstruction of $h$.

\begin{figure}
    \centering 
    \subfloat[trace plot for $\beta_j$]{\label{fig:Gibbs-trace}\includegraphics[width=0.45\textwidth]{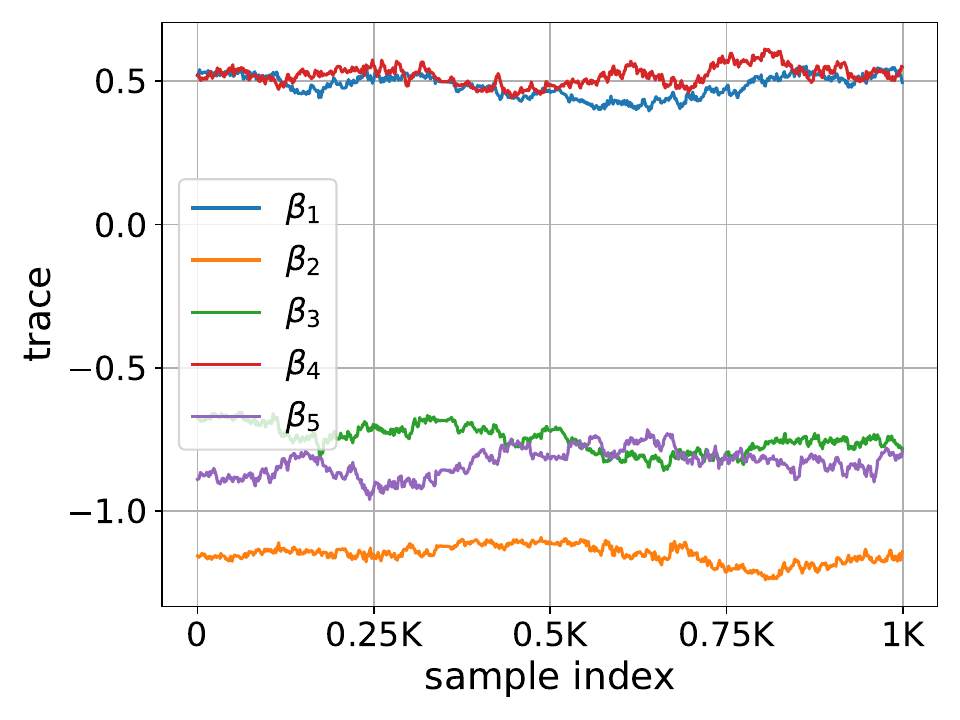}} 
    \subfloat[coefficients]{\label{fig:Gibbs-HPD}\includegraphics[width=0.45\textwidth]{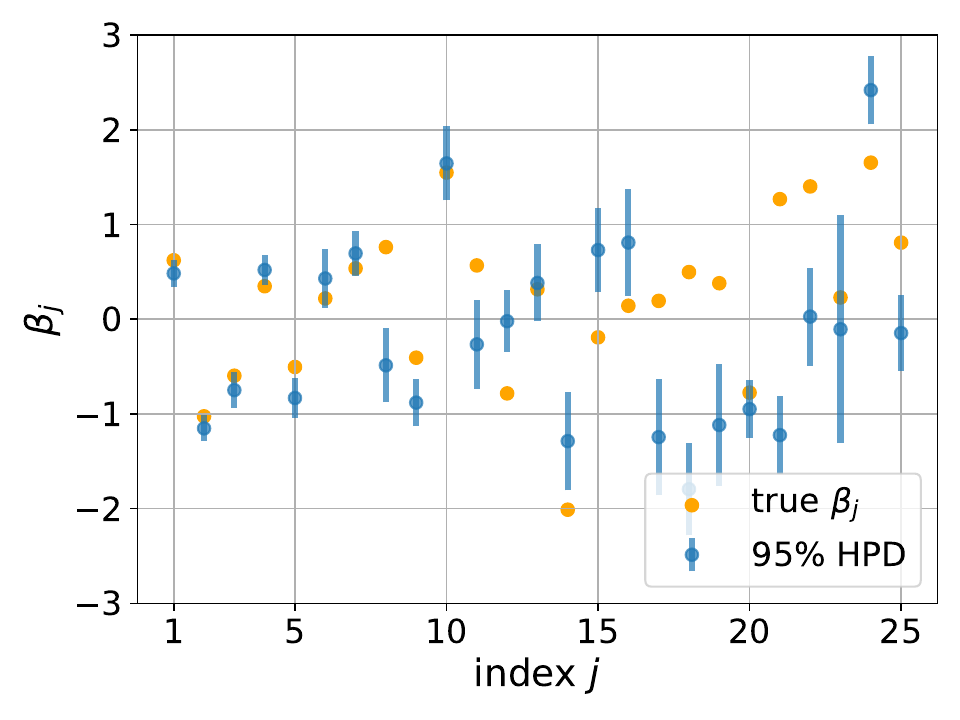}}
    \caption{Diagnostic plots for the Gibbs method. (a) trace plot of the first 5 coefficients $\beta_j$, (b) mean and the highest posterior density interval for each coefficient $\beta_j$.}
    \label{fig:Gibbs-diagnostics}
\end{figure}

The diagnostic plots for the samples of the regularity parameter $s$ is presented in \Cref{fig:Gibbs-diagnostics-s}. We notice that the trace of $s$, in \Cref{fig:Gibbs-trace-s}, shows some correlation between the successive samples, after the warm-up phase. We report that the effective sample size \cite{robert1999monte}, estimated using the ArviZ Python package \cite{arviz_2019}, is approximately 10 for this chain, indicating that roughly 10 effectively independent samples can be extracted from the Markov chain. While this is generally considered a low effective sample size for reliable uncertainty quantification, it may also reflect a complex or potentially multimodal posterior structure. We provide a kernel density estimation (KDE) \cite{fan1996study} of the posterior distribution of $s$ in \Cref{fig:Gibbs-KDE-s}, together with the true value of $s$. We see that the credible values of $s$ are an upper-bound for the true $s$. Recall that $h\in H^{\tau}$ for all $\tau < s-1/2$. Therefore, an estimate of $s$ is in the form of an upper-bound for the true regularity $s$.

\begin{figure}
    \centering 
    \subfloat[trace plot for $s$]{\label{fig:Gibbs-trace-s}\includegraphics[width=0.45\textwidth]{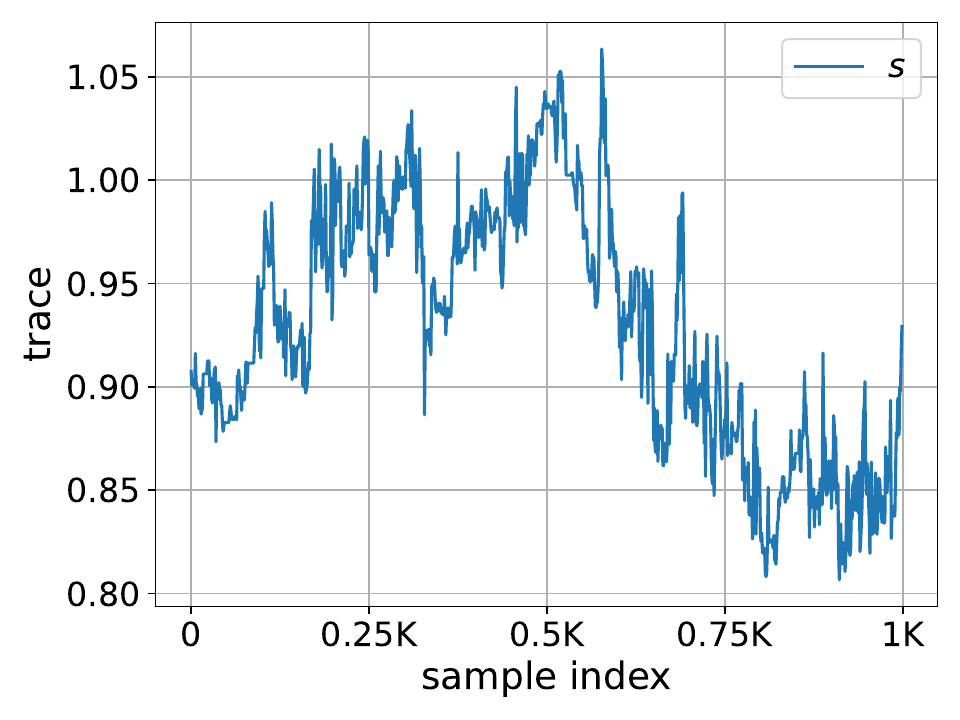}} 
    \subfloat[KDE for $s$]{\label{fig:Gibbs-KDE-s}\includegraphics[width=0.45\textwidth]{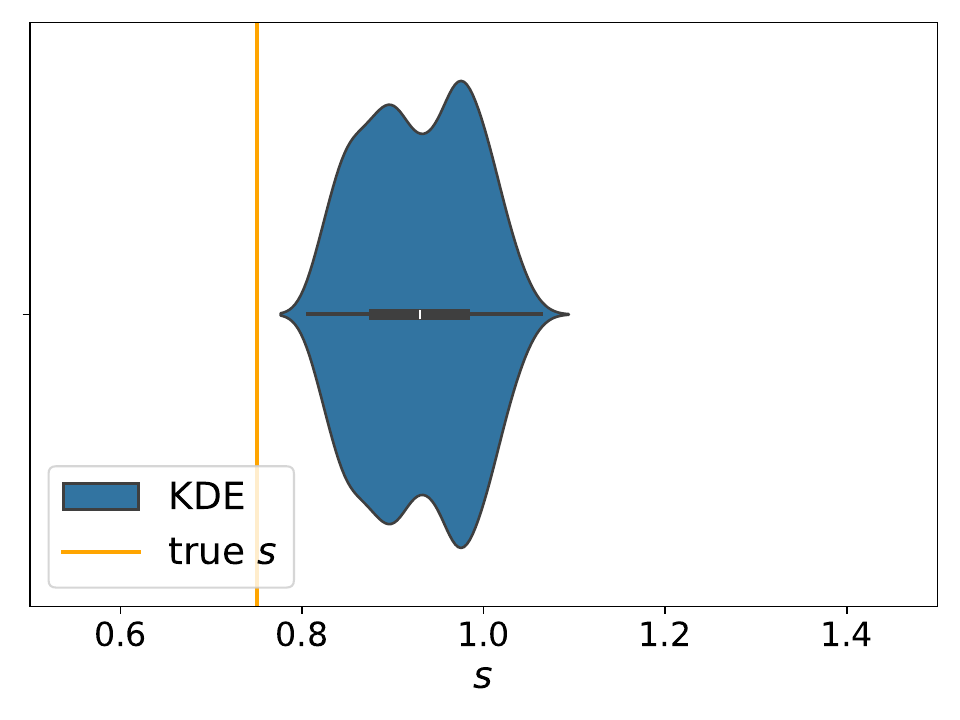}}
    \caption{Diagnostic plots for the Gibbs method. (a) trace plot of the regularity parameter $s$, (b) the kernel density estimation of the posterior samples of $s$ together with the true regularity parameter.}
    \label{fig:Gibbs-diagnostics-s}
\end{figure}

To compute the posterior mean of the seabed, we first fix the regularity parameter to the mean regularity $\bar s := \mathbb E(s)$. We then estimate the seabed following the method discussed in \Cref{sec:fixed-regularity}. The estimated seabed is presented in \Cref{fig:Gibbs-mean-uq} with the quantified uncertainty in the form of a CI band. We notice that estimation is a good match to the true seabed and when the estimated seabed diverges from the true seabed, the width of the CI band increases, indicating the presence of higher uncertainty. However, this width is smaller than the case with a fixed regularity, in \Cref{sec:fixed-regularity}.

\begin{figure}
    \centering 
    \includegraphics[width=0.8\textwidth]{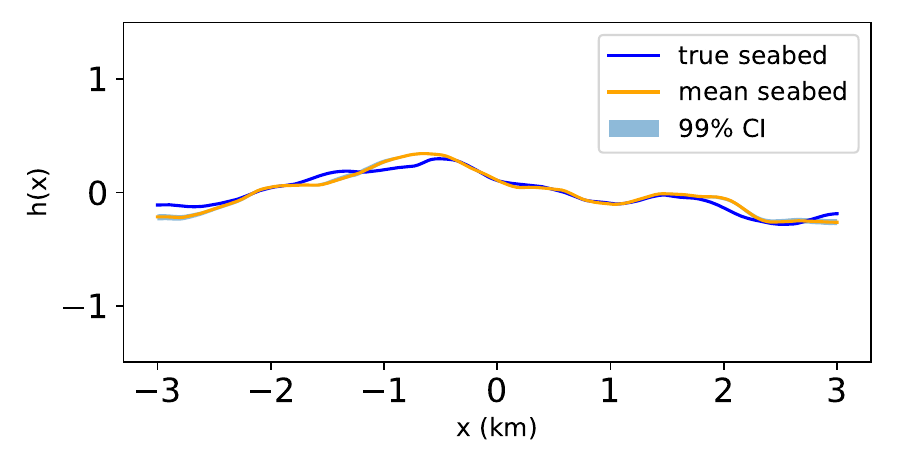}
    \caption{Estimated seabed and the quantified uncertainty with 99\% highest posterior density region.}
    \label{fig:Gibbs-mean-uq}
\end{figure}

\subsection{Uncertain regularity \texorpdfstring{$s$}{Lg} together with an out of prior seabed}

In this section we estimate a seabed that is not a sample from the prior distribution. To create such sample that exhibit homogeneous and isotropic behavior, we follow the idea of constructing a Gaussian random field by filtering (convolving) Gaussian white noise with a smoothing kernel \cite{gelbaum2012white}. We first draw a realization of standard normal random vector $\mathbf p$ of size 64. We then map $\mathbf p$ to the fine grid of size 512 by a periodic linear interpolation operator $I_{\text{inter}}$. We then define periodic Gaussian kernel $k_{\text{ker}}(s) = \frac{1}{z} \exp{(-(s/0.25)^2/2)}$, where $z\in\mathbb R$ is the normalization constant to ensure that the kernel integrates to 1. The mean-zero convolution operation is then consturcted as
\[
    \mathbf h_{\text{true}} = \mathcal B\left( k_{\text{ker}}*(I_{\text{inter}}\mathbf p) - m \right),
\]
where $m\in \mathbb R$ is the average value in $k_{\text{ker}}*(I_{\text{inter}}\mathbf p)$. The operator $\mathcal B$ is a bounding operator, rescaling the signal ensure the range between values [-0.25,0.25]. We report that we apply the convolution using the fast Fourier transform method \cite{brigham1988fast}. The noise vector, interpolated function and smoothened seabed used in this experiment is provided in \Cref{fig:custom-sginal}.

The posterior is constructed and the sampling is performed identically to \Cref{sec:uncertain-regularity}, with $N^h_{\text{inner}}=N^s_{\text{inner}}=250$ and $N_{\text{sample}}=10$ during the warm-up phase, and $N^h_{\text{inner}}=N^s_{\text{inner}}=20$ and $N_{\text{sample}}=500$ after warm-up. We report that step-sizes were tuned to $\beta_h\approx 0.006$ and $\beta_s\approx 0.025$ after the warm-up phase. 

We show the trace plots for the coefficients $\beta_j$ and the regularity parameter $s$ in \Cref{fig:custom-Gibbs-trace,fig:custom-Gibbs-trace-s}, respectively. We see that the chains behave similar to the previous sections. The chain for $s$ shows mild correlation between samples after the warm-up phase. We report that the effective sample size for $s$ in this experiment is approximately 10.

We show the kernel density estimate of the regularity parameter $s$ in \Cref{fig:custom-Gibbs-KDE-s}. We report that the method estimates the mean regularity $\bar s := \mathbb E(s) \approx 0.825 $. We notice that the KDE in \Cref{fig:custom-Gibbs-KDE-s} suggests a tight distribution around the mean.

\begin{figure}
    \centering 
    \subfloat[trace plot]{\label{fig:custom-Gibbs-trace}\includegraphics[width=0.3\textwidth]{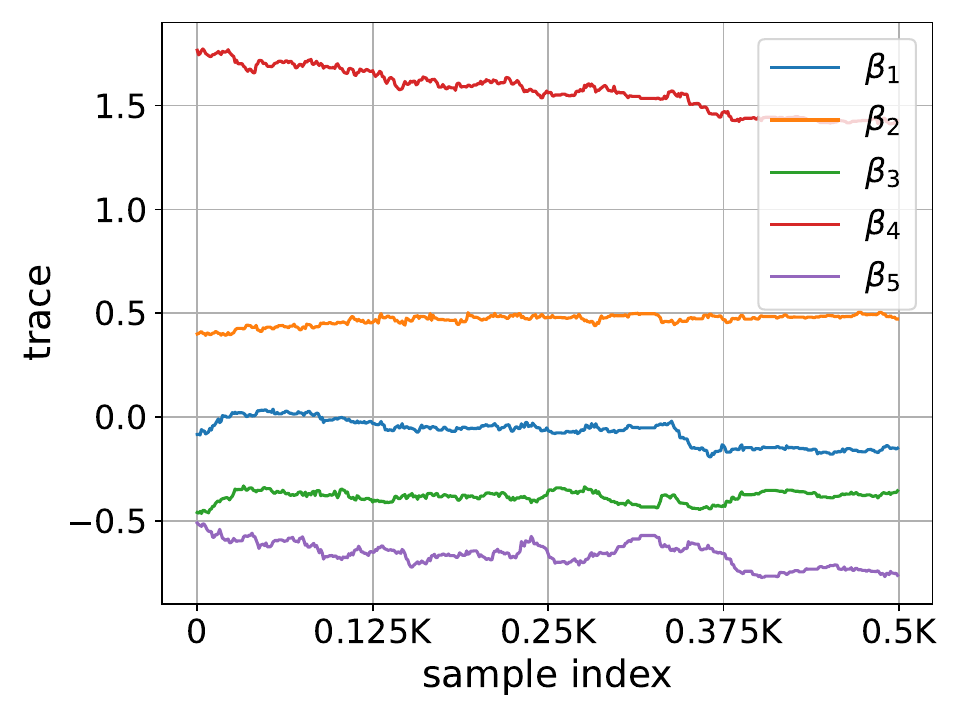}} 
    \subfloat[trace plot]{\label{fig:custom-Gibbs-trace-s}\includegraphics[width=0.3\textwidth]{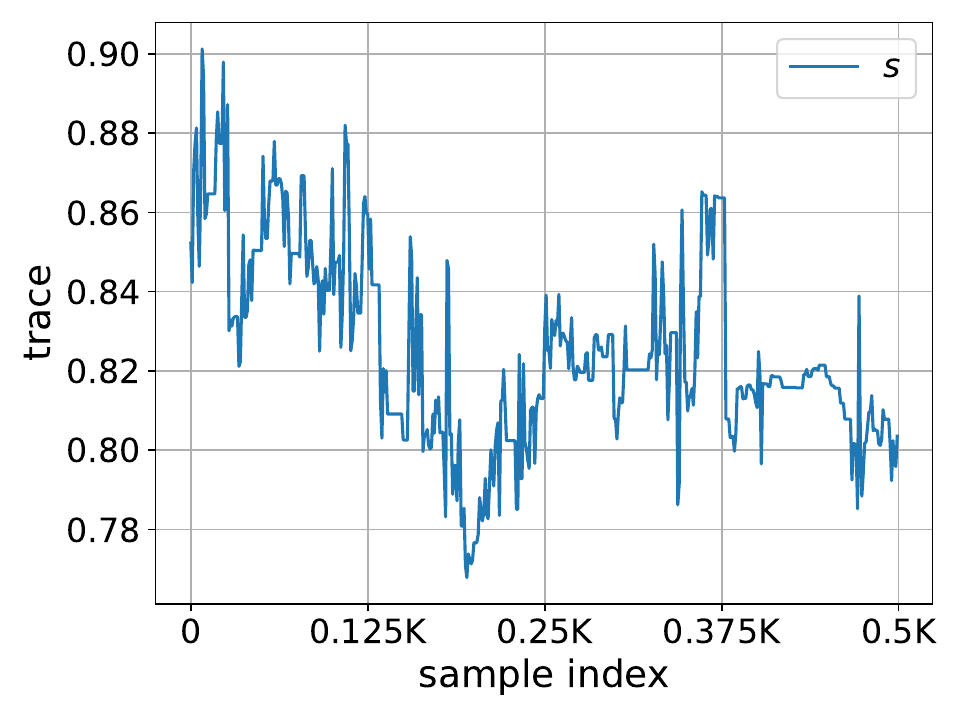}} 
    \subfloat[coefficients]{\label{fig:custom-Gibbs-KDE-s}\includegraphics[width=0.3\textwidth]{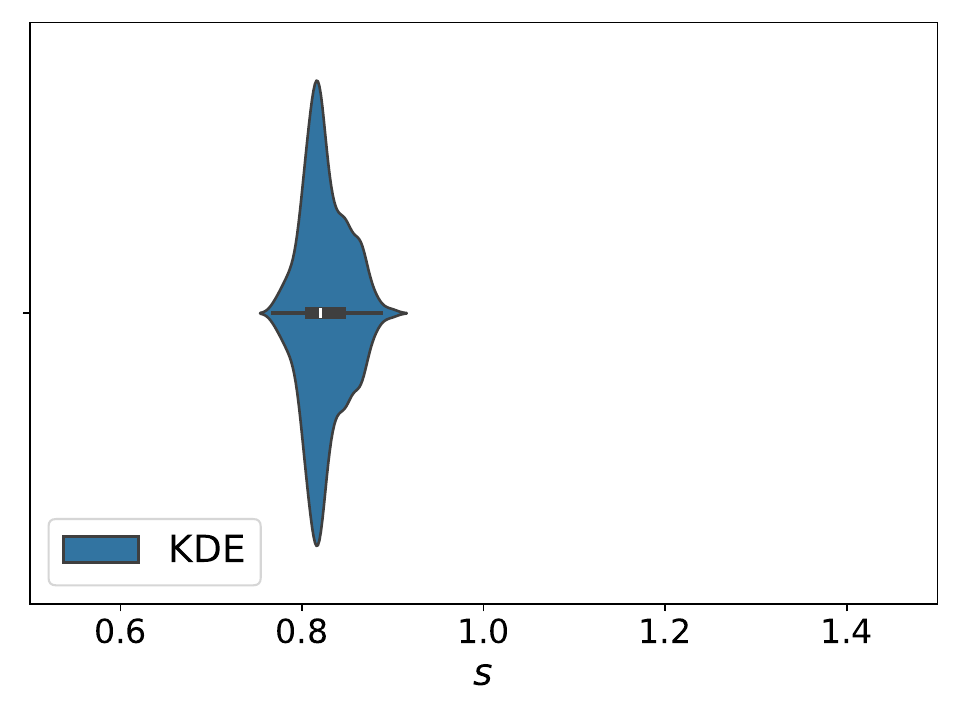}}
    \caption{Diagnostic plots for the Gibbs method. (a) trace plot of the first 5 coefficients $\beta_j$, (b) mean and the highest posterior density interval for each coefficient $\beta_j$.}
    \label{fig:custom-Gibbs-diagnostics}
\end{figure}

We show the estimated mean seabed in \Cref{fig:custom-Gibbs-mean-uq}. We evaluate this estimate following the method discussed in \Cref{sec:uncertain-regularity}. The method can clearly identify the existence of convexity and concavities in the seabed. We notice larger uncertainty at the ends of the seabed, where the estimated seabed diverges from the true seabed. The larger uncertainties are due to the fact that the emitted pressure waves are reflected outside the domain, and therefore, data is incomplete regarding these areas. We remind the reader that the domain of interest in the seabed is roughly the interval $[-2,2]$ where the method provides a mean estimate with a remarkable match to the true seabed.

\begin{figure}
    \centering
    \subfloat[noise realization]{
        \label{fig:custom-sginal}
        \includegraphics[width=0.45\textwidth]{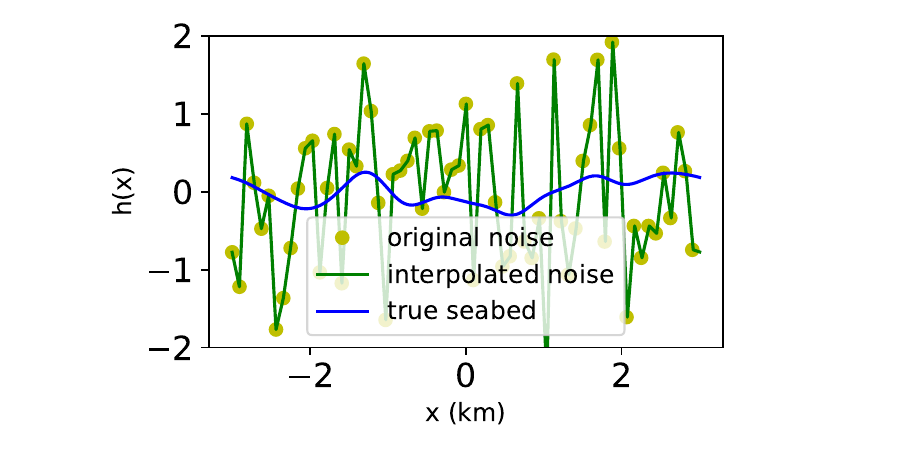}
    }
    \hfill
    \subfloat[seabed estimation]{
        \label{fig:custom-Gibbs-mean-uq}
        \includegraphics[width=0.45\textwidth]{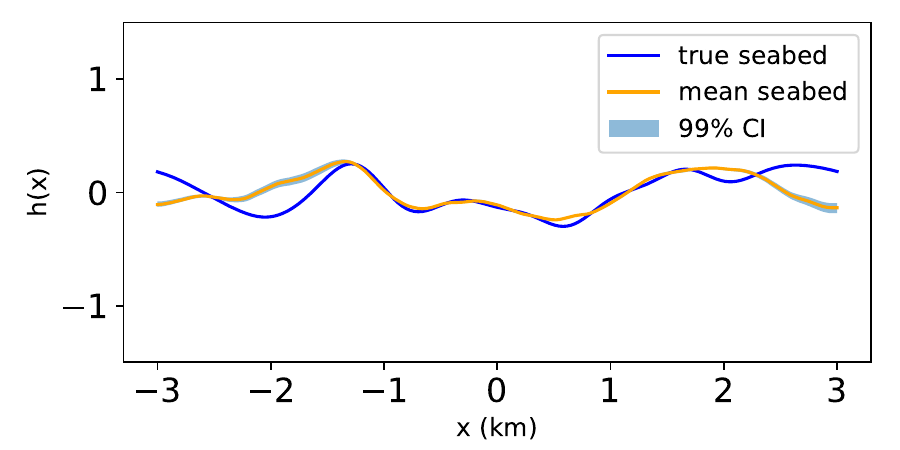}
    }
    \caption{Left: out of prior sample for the seabed $h$. Right: Estimated seabed and quantified uncertainty with 99\% highest posterior density region.}
    \label{fig:custom}
\end{figure}

\section{Conclusions} \label{sec:conclusion}
 
In this paper, we developed a novel goal-oriented Bayesian framework for the inverse problem of seabed tomography. By modeling the seabed as a random function and employing a hierarchical Bayesian approach, we were able to jointly infer both the seabed profile and its regularity from noisy and incomplete measurements. The key innovation lies in interpreting the seabed's roughness in terms of its fractional differentiability, allowing for a more robust and generalizable characterization of surface features beyond the limitations of fixed-frequency or discrete-mode analyses. Our framework utilizes infinite-dimensional priors to construct discretization-invariant models, making it scalable and suitable for large-scale geophysical applications.

We rigorously demonstrated the well-posedness of the resulting posterior distribution, proving its measurability, continuity, and local Lipschitz dependence on data. Additionally, we introduced a Metropolis-within-Gibbs sampling scheme that enables efficient exploration of the joint posterior over the seabed and its regularity. Numerical experiments showed that the method accurately estimates both the shape and the roughness of the seabed, even when the ground truth lies outside the support of the prior distribution. Importantly, the uncertainty quantification provided by the posterior not only highlights regions of high confidence but also identifies areas where data limitations or model mismatch may affect estimation reliability.

These results indicate that the proposed approach provides a promising and scalable framework for interpretable seabed inference with quantified uncertainty. While the present study adopts simplifying assumptions to isolate the role of seabed roughness, it establishes a principled Bayesian methodology for its direct estimation. Future work will focus on extending the framework to more realistic settings, including complex ocean–ground interactions, full elastodynamic forward models, and three-dimensional geometries, as well as on improving prior modeling and developing optimal strategies for uncertainty mitigation. Ultimately, this approach offers a flexible foundation for high-fidelity seabed mapping, with important implications for environmental monitoring, offshore engineering, and hazard mitigation.

\section*{Acknowledgment}

We would like to thank Prof. Per Christian Hansen for his insightful feedback,  careful reading and refinement of the text.

\appendix
\section{Wellposedness}\label{sec:appendix} We detail here a number of theoretical results about well posedness of the forward problem (\ref{eq:wave-eq}) and continuous dependence of the seabed
profile $h$.

Let $\Omega\subset \mathbb R^2$ be a bounded rectangular domain. 
Let us consider the wave model (\ref{eq:wave-eq}) where $c = \sqrt{\alpha/\rho}$
and $y=h(x) \in C([a,b])$ defines a curve contained in $\Omega$
such that $\rho$ and $\alpha$ take different constant values above and below it. We assume
that $h(x)$ meets the edges of $\Omega$ horizontally and that $(x,-1.5)$, $x\in [a,b]$, is the 
bottom boundary of $\Omega$.
%$\partial \Omega^+$ is the region of the boundary where we enforce zero Neumann conditions 
%and $\partial \Omega^-$ the region where we enforce artificial non-reflecting boundary conditions. 

We assume 
\begin{itemize}
\item $\rho, \alpha, c  \in L^\infty(\Omega)$, with $\alpha_{\rm max} > \alpha > \alpha_{\rm min} > 0$, 
$\rho_{\rm max} > \rho > \rho_{\rm min} > 0$, $c > c_{\rm min} >0$,
\item $f(t) \in C^1([0,T])$ and $g(\mathbf x) \in L^2(\Omega)$.
\end{itemize}
% They will never be H^2 Cannot use eigenfunction
We denote by $H^m(\Omega)$ the standard Sobolev spaces, and by $H^s(\Gamma)$,
$\Gamma \subset \partial \Omega$ the spaces of traces, $m,s \in \mathbb R$, see
\cite{lionsmagenes}.
We can construct weak solutions of (\ref{eq:wave-eq}) by means of the variational 
problem: Find $u \in C([0,T]; H^1(\Omega)) \cap C^1([0,T]; L^2(\Omega))$
such that
\begin{equation} \label{eq:varw}
    \begin{aligned}
        \frac{d^2}{dt^2} \int_{\Omega}  \rho u(t)  w  \, d \mathbf x &+ \int_{\Omega}  \alpha \,\nabla  u(t)  \nabla  w  \, d \mathbf x \\
    &+  \frac{d}{dt} \int_{\partial \Omega^-} \frac{\alpha}{c} u(t)  w \, d s =  f(t)  \int_{\Omega}  \rho g  w \,  d \mathbf x, \quad t>0, \\
    \end{aligned}
\end{equation}
with initial conditions $u(\mathbf x, 0) =  0$ and $u_t(\mathbf x, 0) =  0$ and
%\begin{eqnarray} 
%\begin{array}{l} \displaystyle
%\frac{d^2}{dt^2} \int_{\Omega}  \rho u(t)  w  \, d \mathbf x +
%\int_{\Omega}  \alpha \,\nabla  u(t)  \nabla  w  \, d \mathbf x + 
%\frac{d}{dt} \int_{\partial \Omega^-} \frac{\alpha}{c} u(t)  w \, d S_{\mathbf x} 
%\\[2ex] \displaystyle  
%\hfill =  f(t)  \int_{\Omega}  \rho g  w \,  d \mathbf x, \quad t>0,
%\\[2ex] \displaystyle 
%u(0) =  0, \quad u_t(0) =  0, 
%\end{array}  \label{eq:varw}
%\end{eqnarray}
for all $w \in  H^1(\Omega)$ and $t \in (0,T]$.

In the sequel, we denote by $C^m([0,T])$ the spaces of continuous functions with continuous
derivatives up to the order $m$ in the interval $[0,T]$. We use $L^\infty(0,T; B)$ and
$C^m([0,T];B)$ the standard spaces of bounded and continuous functions with continuous
derivatives up to the order $m$ in an interval with values in Banach (or Hilbert) spaces
$B$, in accordance with \cite{lionsmagenes}.

\begin{theorem} \label{thm:A1}
\it Let us assume that 
$\rho, \alpha, c  \in L^\infty(\Omega)$, 
$\alpha > \alpha_{\rm min} > 0$, 
$\rho > \rho_{\rm min} > 0$, 
$c > c_{\rm min} >0$,   
$f(t) \in C^1([0,T])$ and $g(\mathbf x) \in L^2(\Omega)$.
% such that $ g/rho  \in L^2(\Omega)$.
%\footnote{typically the right hand side in the wave equation 
%is $\rho f(\mathbf x,t)$, so it should be $\rho(\mathbf x) g(\mathbf x)$.}.
Here $\alpha_{\text{min}}$, $\rho_{\text{min}}$ and $c_{\text{min}}$ define scalar lower bounds for $\alpha$, $\rho$ and $c$, respectively. Then, there is a unique solution 
$u \in C^1([0,T]; H^1(\Omega)) \cap C^2([0,T]; L^2(\Omega))$
of the variational problem (\ref{eq:varw}).
%{\bf Remark: $u_t \in L^\infty(0,T;L^2(\partial \Omega^-))$} by construction !!!
Moreover,
\begin{equation}
    \begin{aligned}
    \|u_t\|_{L^\infty(0,T;L^2(\Omega))} &\leq
C(T,\rho_{\rm min}, \|f\|_{C([0,T])}, \| g \|_{L^2(\Omega)}) \\
         \| \nabla u \|_{L^\infty(0,T;L^2(\Omega))} &\leq
C(T,\rho_{\rm min}, \alpha_{\rm min}, c_{\rm min}, \|f\|_{C([0,T])}, \| g \|_{L^2(\Omega)}), \\
        \| u_t \|_{L^2(0,T;L^2(\partial \Omega^-))} &\leq 
C(T,\rho_{\rm min}, \alpha_{\rm min}, c_{\rm min},
\|f\|_{L^\infty(0,T)}, \| g \|_{L^2(\Omega)}), \\
        \|  u \|_{L^\infty(0,T;L^2(\Omega))} &\leq C(T, \rho_{\rm min}, \alpha_{\rm min}, c_{\rm min}, 
\mu, \|f\|_{C([0,T])}, \| g \|_{L^\infty(0,T;L^2(\Omega))}), \\
        \| u_{tt}\|_{L^\infty(0,T;L^2(\Omega))} &\leq
C(T,\rho_{\rm min}, \|f\|_{C^1([0,T])}, \| g \|_{L^2(\Omega)}), \\
        \| \nabla u_t \|_{L^\infty(0,T;L^2(\Omega))} &\leq
C(T, \rho_{\rm min}, \alpha_{\rm min}, c_{\rm min}, \|f\|_{C^2([0,T])}, \| g \|_{L^2(\Omega)}), \\
        \| u_{tt} \|_{L^2(0,T;L^2(\partial \Omega^-))} &\leq C(T,\rho_{\rm min}, \alpha_{\rm min}, 
c_{\rm min}, \|f\|_{C^2([0,T])}, \| g \|_{L^2(\Omega)}),
    \end{aligned}
\end{equation}
%\[\| u_t\|_{L^\infty(0,T;L^2(\Omega))} \leq
%C(T,\rho_{\rm min}, \|f\|_{C([0,T])}, \| g \|_{L^2(\Omega)}),\]
%\[ \| \nabla u \|_{L^\infty(0,T;L^2(\Omega))} \leq
%C(T,\rho_{\rm min}, \alpha_{\rm min}, c_{\rm min}, %\|f\|_{C([0,T])}, \| g \|_{L^2(\Omega)}),\]
%\[ \| u_t \|_{L^2(0,T;L^2(\partial \Omega^-))} \leq 
%C(T,\rho_{\rm min}, \alpha_{\rm min}, c_{\rm min},
%\|f\|_{L^\infty(0,T)}, \| g \|_{L^2(\Omega)}),\]
%\[ \|  u \|_{L^\infty(0,T;L^2(\Omega))} \leq C(T, \rho_{\rm min}, \alpha_{\rm min}, c_{\rm min}, 
%\mu, \|f\|_{C([0,T])}, \| g \|_{L^\infty(0,T;L^2(\Omega))}),\]
%\[\| u_{tt}\|_{L^\infty(0,T;L^2(\Omega))} \leq
%C(T,\rho_{\rm min}, \|f\|_{C^1([0,T])}, \| g \|_{L^2(\Omega)}),\]
%\[ \| \nabla u_t \|_{L^\infty(0,T;L^2(\Omega))} \leq
%C(T, \rho_{\rm min}, \alpha_{\rm min}, c_{\rm min}, %\|f\|_{C^2([0,T])}, \| g \|_{L^2(\Omega)}),\]
%\[ \| u_{tt} \|_{L^2(0,T;L^2(\partial \Omega^-))} \leq %C(T,\rho_{\rm min}, \alpha_{\rm min}, 
%c_{\rm min}, \|f\|_{C^2([0,T])}, \| g \|_{L^2(\Omega)}),\]
for some $\mu>0$, that is, the norms of the solutions depend continuously
on $f$, $g$, $\rho_{\rm min}$ and $\alpha_{\rm min}$.
\end{theorem}

\begin{pf}
The construction method uses countable bases of $H^1(\Omega)$ and follows 
\cite{lionsmagenes}, Chapter 3, for standard wave equations with a number of modifications in 
view of the nonstandard boundary condition $\frac{\partial u}{\partial \mathbf n} 
= - u_t$, see \cite{ABUGATTAS2025129453} for details.
Notice that $H^1(\Omega)$ is separable, thus we can always find a set 
$\{\phi_1, \ldots, \phi_k, \ldots \}  \subset H^1(\Omega)$ whose elements
are linearly independent  while their linear combinations with real coefficients are 
dense  in $H^1(\Omega)$, see \cite{lionsmagenes}.
\end{pf}

We next establish continuous dependence of the solutions on the seafloor profile 
$h$ for coefficients $\rho$ and $\alpha$ given by (\ref{coefficients}),

\begin{theorem} \label{thm:A2}
\it Under the hypotheses of Theorem 1,
let us further  assume that $f(t) \in C^2([0,T])$. 
Then, the solution to (\ref{eq:wave-eq}) has regularity
$C^2([0,T]; H^1(\Omega)) \cap C^3([0,T]; L^2(\Omega))$ and
satisfies the additional estimates
\begin{eqnarray} \begin{array}{l}
\| u_{ttt}\|_{L^\infty(0,T;L^2(\Omega))} \leq
C(T,\rho_{\rm min}, \|f\|_{C^2([0,T])}, \| g \|_{L^2(\Omega)}), 
\\[2ex] \displaystyle
\| \nabla u_{tt} \|_{L^\infty(0,T;L^2(\Omega))} \leq
C(T,\rho_{\rm min}, \alpha_{\rm min}, c_{\rm min}, \|f\|_{C^2([0,T])}, 
\| g \|_{L^2(\Omega)}), 
 \\[2ex] \displaystyle
\| u_{ttt} \|_{L^2(0,T;L^2(\partial \Omega^-))} \leq C(T,\rho_{\rm min}, 
\alpha_{\rm min}, c_{\rm min}, \|f\|_{C^2([0,T])}, \| g \|_{L^2(\Omega)} ).
\end{array} \label{est_2} \end{eqnarray}
Moreover, the solutions of (\ref{eq:wave-eq}) depend continuously
on the curve $h$ defining $\alpha$, with respect to the 
$L^2$ norms.  More precisely, given continuous curves
$y=h_1(x)$ and $y=h_2(x)$, the corresponding solutions
$u_{h_1}$ and $u_{h_2}$ satisfy that
$\|u_{h_1} - u_{h_2} \|_{H^1(\Omega)}^2$, 
$\|u_{h_1,t} - u_{h_2,t} \|_{L^2(\Omega)}^2$ and
$\|u_{h_1} - u_{h_2} \|_{L^2(\partial \Omega^+)}$
are bounded from above
by $\| \alpha_{h_1}  -  \alpha_{h_2} \|_{L^2(\Omega)} 
+\| \alpha_{h_1}  -  \alpha_{h_2} \|_{L^2(\partial \Omega^-)}$
multiplied by constants not depending on $h_1$ and $h_2$. 
\end{theorem}

\begin{pf} We differentiate twice with respect to time (\ref{eq:wave-eq})
formulated as (\ref{eq:varw})
to get a similar boundary value problem for $u_{tt}$ with initial data
$u_{tt}(\mathbf x,0)= f(0) g(\mathbf x) \in H^1(\Omega)$ and
$u_{ttt}(\mathbf x,0)= f'(0) g(\mathbf x) \in L^2(\Omega)$ and
% divide by rho both
right hand side $f'' g \in C([0,T];L^2(\Omega))$. This yields the augmented 
regularity, i.e., $u_{tt}$ inherits the regularity according to Theorem 1 and we get the stated bounds.
As a consequence, for each $t$ the function $u_t$ is a solution of an 
elliptic boundary value problem 
$- \nabla \cdot (\alpha \nabla u_t) = - \rho u_{ttt} +  \rho f' g \in L^2(\Omega)$
with zero Neumann boundary conditions on $\partial \Omega^+$
and $\nabla u_t \cdot \mathbf n = - u_{tt}/c$ in $L^2(\partial \Omega^-).$

We achieve $H^2$ regularity near the boundaries as follows. The boundary
$\partial \Omega$ being compact, we can find a finite set of open sets $U_j$, 
$j=1,\ldots,N$ covering it. In this case, we can take $N=6$, and open sets 
such that $U_j \cap \Omega$ is fully contained in the regions $y>h(x)$ or 
$y<h(x)$, with boundaries each of the $6$ associated segments of 
$\partial \Omega$.
Consider a function $\eta \in C^\infty(U_j)$ vanishing on a layer around
$\partial U_j \cap \Omega$ but not on $U_j \cap \partial \Omega$.
Then,  $w = u _t\eta$ satisfies
\begin{eqnarray*}\begin{array}{ll}
\alpha \Delta w = \eta [- \rho u_{ttt} +  \rho f' g] + 2 \alpha \nabla u_t \nabla \eta 
+ \alpha u _t \Delta \eta = q, & \mathbf x \in U_j \cap \Omega, \\ [1ex]
w = \nabla w \cdot \mathbf n = 0, & \mathbf x \in \partial U_j \cap \Omega, 
\end{array}\end{eqnarray*}
and either $\nabla w \cdot \mathbf n = - \eta  u_{tt}/c + u \nabla \eta \cdot
\mathbf n$, $ \mathbf x \in \partial \Omega^- \cap U_j $ or
$\nabla w \cdot \mathbf n = 0$, $\mathbf x \in \partial \Omega^+ \cap U_j $.
Given piecewise $C^1$ sets $R$, for any $g \in L^2(\partial R)$, there exists a 
function $w_g \in H^2(U_j \cap \Omega)$ such that $\nabla w_g \cdot \mathbf n 
= g$. Setting $w = \tilde w + w_g$, we find
\begin{eqnarray*}\begin{array}{l}
\alpha \Delta \tilde w(t)  =  q  -  \alpha \Delta w_g(t) \in L^2(U_j \cap \Omega), 
\\ [1ex]
\nabla \tilde w(t) \cdot \mathbf n = 0,  \quad {\rm on }  \, \partial (U_j \cap \Omega). 
\end{array}\end{eqnarray*}
By the regularity results in \cite{grisvard}, 
Ch. 3, the solution $w \in H^2(U_j \cap \Omega)$.  By Sobolev's injections 
\cite{adams}, in dimension $n=2$, $u$ is continuous up to the borders, with
possible exception of merging points.
%and takes on $L^2$ Dirichlet boundary conditions on $\partial \Omega$
%and on the curve $y=h(x)$. 
%{\bf FALSE Then,
%elliptic regularity \cite{grisvard} implies that $u_t(t)$ has $H^2$
%regularity in each of the regions $y>h(x)$  and $y<h(x)$. We will
%need this regularity for the continuous dependence proof.}

Now, let $u_1$ and $u_2$ be two solutions in 
$C^2([0,T]; H^1(\Omega)) \cap C^3([0,T]; L^2(\Omega))$
to (\ref{eq:wave-eq}) corresponding to coefficients $\alpha_{h_1}$ 
and $\alpha_{h_2}$ defined by curves $y=h_1(x)$ and $y=h_2(x)$.
We set $u=u_1-u_2$.  Notice that $c_{h_i} = \sqrt{ \alpha_{h_i}/\rho}$
and $\alpha_{h_i}/c_{h_i} = \sqrt{\rho \alpha_{h_i}}$, $i=1,2$.
Then, $u$ is a solution to
\begin{equation*}
    \begin{aligned}
        &\frac{d^2}{dt^2} \int_{\Omega}  \rho u(t)  w  \, d \mathbf x +
\int_{\Omega}  \alpha_{h_1} \,\nabla  u(t)  \nabla  w  \, d \mathbf x + 
 \frac{d}{dt} \int_{\partial \Omega^-}  \sqrt{\rho} \sqrt{\alpha_{h_1}}
u(t)  w \, d s = \\
    & - \int_{\Omega}  (\alpha_{h_1} -  \alpha_{h_2}) \,
\nabla  u_2(t)  \nabla  w  \, d \mathbf x  -
 \int_{\partial \Omega^-}  \sqrt{\rho} (\sqrt{\alpha_{h_1}}-
\sqrt{\alpha_{h_2}}) u_{2,t}(t)  w \, d s  
    \end{aligned}
\end{equation*}

%\begin{eqnarray*} 
%\begin{array}{l} \displaystyle
%\frac{d^2}{dt^2} \int_{\Omega}  \rho u(t)  w  \, d \mathbf x +
%\int_{\Omega}  \alpha_{h_1} \,\nabla  u(t)  \nabla  w  \, d \mathbf x + 
% \frac{d}{dt} \int_{\partial \Omega^-}  \sqrt{\rho} \sqrt{\alpha_{h_1}}
%u(t)  w \, d s =
% \\[2ex]
%\displaystyle
%- \int_{\Omega}  (\alpha_{h_1} -  \alpha_{h_2}) \,
%\nabla  u_2(t)  \nabla  w  \, d \mathbf x  -
% \int_{\partial \Omega^-}  \sqrt{\rho} (\sqrt{\alpha_{h_1}}-
%\sqrt{\alpha_{h_2}}) u_{2,t}(t)  w \, d s  
%\end{array}  \label{eq:diff}
%\end{eqnarray*}
for all $w \in  H^1(\Omega)$ and $t \in [0,T]$, with $u(0) = u_t(0) = 0$. 
We set $u= e^{\mu t} v$, $\mu>0$, and write down the equation satisfied
by $v$, as before. Setting $w=v_{t}$, we get 
\begin{equation*}
\begin{aligned}
    &\frac{1}{2}\frac{d}{dt} \int_{\Omega} \rho |v_t(t)|^2  d \mathbf x  
+ \int_{\partial \Omega^-}   \sqrt{\rho} \sqrt{\alpha_{h_1}} |v_t(t)|^2  
d s
+ 2 \mu \int_{\Omega} \rho |v_t(t)|^2  d \mathbf x  + \\
    &\frac{1}{2}\frac{d}{dt}\left [ \int_{\Omega} \alpha_{h_1} |\nabla v(t)|^2  d
 \mathbf x  
+  \mu^2 \int_{\Omega}  \rho  |v(t)|^2    \, d \mathbf x
+  \mu \int_{\partial \Omega^-}  \sqrt{\rho} \sqrt{\alpha_{h_1}}
 |v(t)|^2    \, d s
\right] = \\
&\int_{\Omega} e^{-\mu t} (\alpha_{h_2} \! -  \! \alpha_{h_1}) \,
\nabla  u_2(t)  \nabla  v_t(t)  \, d \mathbf x  -
 \int_{\partial \Omega^-} \hskip -4mm
  e^{-\mu t} \sqrt{\rho} (\sqrt{\alpha_{h_1}} \! - \!
  \sqrt{\alpha_{h_2}}) u_{2,t}(t) v_t(t)  \, d s  
\end{aligned}
\end{equation*}

%\begin{eqnarray*} 
%\begin{array}{l} \displaystyle
%\frac{1}{2}\frac{d}{dt} \int_{\Omega} \rho |v_t(t)|^2  d \mathbf x  
%+ \int_{\partial \Omega^-}   \sqrt{\rho} \sqrt{\alpha_{h_1}} |v_t(t)|^2  
%d s
%+ 2 \mu \int_{\Omega} \rho |v_t(t)|^2  d \mathbf x  +
%\\[2ex]  \displaystyle 
%\frac{1}{2}\frac{d}{dt}\left [ \int_{\Omega} \alpha_{h_1} |\nabla v(t)|^2  d
% \mathbf x  
%+  \mu^2 \int_{\Omega}  \rho  |v(t)|^2    \, d \mathbf x
%+  \mu \int_{\partial \Omega^-}  \sqrt{\rho} \sqrt{\alpha_{h_1}}
% |v(t)|^2    \, d s
%\right] = 
%\\[2ex]  \displaystyle 
%\int_{\Omega} e^{-\mu t} (\alpha_{h_2} \! -  \! \alpha_{h_1}) \,
%\nabla  u_2(t)  \nabla  v_t(t)  \, d \mathbf x  -
% \int_{\partial \Omega^-} \hskip -4mm
%  e^{-\mu t} \sqrt{\rho} (\sqrt{\alpha_{h_1}} \! - \!
%  \sqrt{\alpha_{h_2}}) u_{2,t}(t) v_t(t)  \, d s  
%\end{array} 
%\end{eqnarray*}
We have the identity
\begin{equation*}
    \begin{aligned}
            &\int_{\Omega} e^{-\mu t} (\alpha_{h_2} \! -  \! \alpha_{h_1}) \,
            \nabla  u_2(t)  \nabla  v_t(t)  \, d \mathbf x  = \frac{d}{dt}
            \int_{\Omega} e^{-\mu t} (\alpha_{h_2} \! -  \! \alpha_{h_1}) \,
            \nabla  u_2(t)  \nabla  v(t)  \, d \mathbf x \\
            &+ \mu \int_{\Omega} e^{-\mu t} (\alpha_{h_2} \! -  \! \alpha_{h_1}) \,
            \nabla  u_2(t)  \nabla  v(t)  \, d \mathbf x 
            - \int_{\Omega} e^{-\mu t} (\alpha_{h_2} \! -  \! \alpha_{h_1}) \,
            \nabla  u_{2,t}(t)  \nabla  v(t)  \, d \mathbf x.
    \end{aligned}
\end{equation*}

%\begin{eqnarray*}
% \int_{\Omega} e^{-\mu t} (\alpha_{h_2} \! -  \! \alpha_{h_1}) \,
%\nabla  u_2(t)  \nabla  v_t(t)  \, d \mathbf x  = \frac{d}{dt}
%\int_{\Omega} e^{-\mu t} (\alpha_{h_2} \! -  \! \alpha_{h_1}) \,
%\nabla  u_2(t)  \nabla  v(t)  \, d \mathbf x \\
%+ \mu \int_{\Omega} e^{-\mu t} (\alpha_{h_2} \! -  \! \alpha_{h_1}) \,
%\nabla  u_2(t)  \nabla  v(t)  \, d \mathbf x  
%- \int_{\Omega} e^{-\mu t} (\alpha_{h_2} \! -  \! \alpha_{h_1}) \,
%\nabla  u_{2,t}(t)  \nabla  v(t)  \, d \mathbf x.
%\end{eqnarray*}
Integrating we find
\begin{equation} \label{eq:energy1}
    \begin{aligned}
        &\frac{1}{2} \int_{\Omega}  \rho |v_t(t)|^2  d \mathbf x  
         + \int_0^t \int_{\partial \Omega^-}   \sqrt{\rho} \sqrt{\alpha_{h_1}} |v_t(\tau)|^2  
        d s d\tau \\
        & + 2 \mu \int_0^t  \int_{\Omega} \rho |v_t(\tau)|^2  d \mathbf x  d\tau +    \frac{1}{2}\left [ \int_{\Omega} \alpha_{h_1} |\nabla v(t)|^2  d  \mathbf x \right. \\  
        & \left.+  \mu^2 \int_{\Omega}  \rho  |v(t)|^2    \, d \mathbf x +  \mu \int_{\partial \Omega^-}  \sqrt{\rho} \sqrt{\alpha_{h_1}} |v(t)|^2    \, d s \right] = \\
        &   \int_{\Omega} e^{-\mu t} (\alpha_{h_2} \! -  \! \alpha_{h_1}) \,
        \nabla  u_2(t)  \nabla  v(t)  \, d \mathbf x \\  
        &- \int_0^t  \int_{\partial \Omega^-} \hskip -4mm
        e^{-\mu s}  \rho^{1/2}  ( \alpha_{h_1}^{1/2} \! - \!
        \alpha_{h_2}^{1/2} ) u_{2,t}(\tau) v_t(\tau)  \, d s  d\tau \\
        & +  \int_0^t  \int_{\Omega} e^{-\mu \tau}  (\alpha_{h_2} \! -  \! \alpha_{h_1}) \, 
        (\mu -\nabla  u_{2,t}(t)  )\nabla  v(\tau)  \, d \mathbf x d\tau  \\
        &= I_1 + \int_0^t I_2(\tau) d\tau 
        + \int_0^t I_3(\tau) d\tau.
    \end{aligned}
\end{equation}

%\begin{eqnarray} 
%\begin{array}{l} \displaystyle
%\frac{1}{2} \int_{\Omega} \rho |v_t(t)|^2  d \mathbf x  
%+ \int_0^t \int_{\partial \Omega^-}   \sqrt{\rho} \sqrt{\alpha_{h_1}} |v_t(\tau)|^2  
%d s d\tau
%+ 2 \mu \int_0^t  \int_{\Omega} \rho |v_t(\tau)|^2  d \mathbf x  d\tau +
%\\[2ex]  \displaystyle 
%\frac{1}{2}\left [ \int_{\Omega} \alpha_{h_1} |\nabla v(t)|^2  d
% \mathbf x  
%+  \mu^2 \int_{\Omega}  \rho  |v(t)|^2    \, d \mathbf x
%+  \mu \int_{\partial \Omega^-}  \sqrt{\rho} \sqrt{\alpha_{h_1}}
% |v(t)|^2    \, d s 
%\right] = 
%\\[2ex]  \displaystyle 
%  \int_{\Omega} e^{-\mu t} (\alpha_{h_2} \! -  \! \alpha_{h_1}) \,
%\nabla  u_2(t)  \nabla  v(t)  \, d \mathbf x  
%- \int_0^t  \int_{\partial \Omega^-} \hskip -4mm
%  e^{-\mu s}  \rho^{1/2}  ( \alpha_{h_1}^{1/2} \! - \!
%  \alpha_{h_2}^{1/2} ) u_{2,t}(\tau) v_t(\tau)  \, d s  d\tau
%\\[2ex]  \displaystyle 
%+  \int_0^t  \int_{\Omega} e^{-\mu \tau}  (\alpha_{h_2} \! -  \! \alpha_{h_1}) \, 
%(\mu -\nabla  u_{2,t}(t)  )\nabla  v(\tau)  \, d \mathbf x d\tau 
%= I_1 + \int_0^t I_2(\tau) d\tau 
%+ \int_0^t I_3(\tau) d\tau.
%\end{array}  \label{eq:energy1}
%\end{eqnarray}
Using Young's inequality and Sobolev embeddings \cite{brezis} the integrals in the 
right hand side are bounded as follows. First, 
\begin{eqnarray}
\begin{array}{l}
|I_1| \leq 
C_\epsilon\int_{\Omega} | (\alpha_{h_1} - \alpha_{h_2})|^2 |\nabla u_2|^2 d\mathbf x 
+ \epsilon \| \nabla v \|_2^2  \leq  
\\[2ex]  \displaystyle 
C_\epsilon\left( \int_{\Omega} | \alpha_{h_1} - \alpha_{h_2}|^{2q'} \right)^{1/q'} 
\left( \int_{\Omega} |\nabla u_2|^{2q} \right)^{1/q} 
+ \epsilon \| \nabla v \|_2^2,
\end{array} 
\label{eq:i1_1}
\end{eqnarray}
with $\frac{1}{q} + \frac{1}{q'} = 1$, $2q > 2$, with $q>1$ and $q'=\frac{q}{q-1}$.
By Sobolev's injections, $H^1(\Omega) \subset L^q$, $q \leq q^*$, $q^* = \frac{2 n}{n-2}$
if $n>2$ and $H^1(\Omega) \subset L^q$, $2 \leq q <\infty$ if $n=2$. Here, $n=2$, thus
we can take any $1< q < \infty$. Notice that
\begin{eqnarray}
\begin{array}{l}
\int_{\Omega} | (\alpha_{h_1} - \alpha_{h_2})|^{2q'}  d \mathbf x \leq
\| \alpha_{h_1} - \alpha_{h_2} \|_{\infty}^{2q'-2} 
\int_{\Omega} | (\alpha_{h_1} - \alpha_{h_2})|^{2} d \mathbf x
\\[2ex]  \displaystyle 
\leq \alpha_{\rm max}^{2q'-2} \int_{\Omega} | \alpha_{h_1} - \alpha_{h_2}|^{2} d \mathbf x. 
\end{array} \label{eq:aux}
\end{eqnarray}
Inserting (\ref{eq:aux}) in (\ref{eq:i1_1}), we find
\begin{eqnarray}
\begin{array}{l}
|I_1(t)| \leq 
C_{\epsilon_1} \alpha_{\rm max}^{(2q'-2)/q'} \|\alpha_{h_1} - \alpha_{h_2}\|_2^{2/q'}
\| \nabla u_2 (t) \|_{2q}^2  + \epsilon_1 \| \nabla v (t)\|_2^2, \quad t>0, 
\end{array} \label{eq:i1_2}
\end{eqnarray}
where $C_{\epsilon_1} >0$ is adjusted to each $\varepsilon_1 >0$\cite{brezis}.
In a similar way, we find
\begin{eqnarray}
\begin{array}{l}
 | I_2(s)|  \leq C_{\epsilon_2} \rho_{\rm max}^{1/2} 
 \alpha_{\rm max}^{(2q'-2)/2q'} 
\| \alpha_{h_1}^{1/2}   -  \alpha_{h_2}^{1/2}  \|_{L^2(\partial \Omega^-)}^{2/q'}
  \| u_{2,t}(t) \|_{L^{2q}(\partial \Omega^-)}^2
\\[2ex]  \displaystyle   
+  \epsilon_2 \rho_{\rm max}^{1/2} \|v_t(s)  \|_{L^2(\partial \Omega^-)}^2,
\end{array} \label{eq:i2_1}
\end{eqnarray}
\begin{eqnarray}
\begin{array}{l}
 | I_3(s)|   \leq C_{\epsilon_3}  \rho_{\rm max}^{1/2} 
\alpha_{\rm max}^{(2q'-2)/q'} \|\alpha_{h_1} - \alpha_{h_2}\|_2^{2/q'}
\| \mu -\nabla  u_{2,t}(s) \|_{2q}^2 
\\[2ex]  \displaystyle  
+ \epsilon_3 \rho_{\rm max}^{1/2} \| \nabla v(s)\|_2^2.
\end{array} \label{eq:i3}
\end{eqnarray}
Since $\alpha_{h_1}$ and $\alpha_{h_2}$ remain bounded away from zero, we can write
\begin{eqnarray}
\| \alpha_{h_1}^{1/2}   -  \alpha_{h_2}^{1/2}  \|_{L^2(\partial \Omega^-)}
\leq C(\alpha_{\rm min}, \alpha_{\rm max})  \| \alpha_{h_1}  -  \alpha_{h_2} \|_{L^2(\partial \Omega^-)}.
\label{eq:i2_3}
\end{eqnarray}
Since $u_{2,t}$ has $H^2$ regularity in neighboorhods (contained in $y>h_2(x)$ or $y<h_2(x)$) of the $6$ boundary segments $u_{2,t}$ belongs
to $W^{1,2q}$ in them by Sobolev's embeddings and has a $L^{2q}(\partial \Omega^-)$ 
traces bounded by the norms $\|u_{2, t}\|_{2}$, $\|\nabla u_{2, t}\|_2$, $\|u_{2, ttt}\|_2$ and $\| f'g \|_2$:
\begin{eqnarray}
\| u_{2,t}(t) \|_{L^{2q}(\partial \Omega^-)}
\leq C(\Omega, h_2, \|u_{2, t}\|_{2},  \|\nabla u_{2, t}\|_2,  \|u_{2, ttt}\|_2,  \| f'g \|_2).
\label{eq:i2_4}
\end{eqnarray}

Putting together (\ref{eq:energy1})-(\ref{eq:i2_4}) and (\ref{est_2}) we get
\begin{equation} \label{eq:energy2}
    \begin{aligned}
        &\frac{\rho_{\min}}{4} \int_{\Omega}|v_t(t)|^2  d \mathbf x  
        + \frac{1}{2}(\rho_{\rm min} \alpha_{\rm min})^{1/2}
        \int_0^t \int_{\partial \Omega^-}  |v_t(\tau)|^2  d s d\tau \\
        & + 2 \mu \rho_{\rm min} \int_0^t  \int_{\Omega} |v_t(\tau)|^2  d \mathbf x  d\tau + \frac{1}{2}\left[  \alpha_{\rm min} \int_{\Omega} |\nabla v(t)|^2  d \mathbf x   \right. \\
        &\left.+  \mu^2 \rho_{\rm min} \int_{\Omega}  |v(t)|^2    \, d \mathbf x +  \mu (\rho_{\rm min} \alpha_{\rm min})^{1/2} \int_{\partial \Omega^-}   |v(t)|^2    \, d s \right] = \\
        &  C_1 \|\alpha_{h_1} - \alpha_{h_2}\|_{L^2(\Omega)}^{2/q'} +
        C_2 \| \alpha_{h_1}  -  \alpha_{h_2} \|_{L^2(\partial \Omega^-)}^{2/q'} \\
        &+  \epsilon_3  \rho_{\rm max}^{1/2}
        \int_0^t  \int_{\Omega} | \nabla v(\tau) |^2 d \mathbf x d\tau.
        \end{aligned}
\end{equation}

%\begin{eqnarray} 
%\begin{array}{l} \displaystyle
%\frac{\rho_{\min}}{4} \int_{\Omega}|v_t(t)|^2  d \mathbf x  
%+ \frac{1}{2}(\rho_{\rm min} \alpha_{\rm min})^{1/2}
%\int_0^t \int_{\partial \Omega^-}  |v_t(\tau)|^2  d s d\tau
%+ 2 \mu \rho_{\rm min} \int_0^t  \int_{\Omega} |v_t(\tau)|^2  d \mathbf x  d\tau  
%\\[2ex]  \displaystyle 
%+ \frac{1}{2}\left [  \alpha_{\rm min} \int_{\Omega} |\nabla v(t)|^2  d \mathbf x  +  \mu^2 %\rho_{\rm min} \int_{\Omega}  |v(t)|^2    \, d \mathbf x
%+  \mu (\rho_{\rm min} \alpha_{\rm min})^{1/2} \int_{\partial \Omega^-}   |v(t)|^2    \, d s
%\right] = 
%\\[2ex]  \displaystyle 
% C_1 \|\alpha_{h_1} - \alpha_{h_2}\|_{L^2(\Omega)}^{2/q'} +
%5 C_2 \| \alpha_{h_1}  -  \alpha_{h_2} \|_{L^2(\partial \Omega^-)}^{2/q'} 
%+  \epsilon_3  \rho_{\rm max}^{1/2}
%\int_0^t  \int_{\Omega} | \nabla v(\tau) |^2 d \mathbf x d\tau.
%\end{array}  \label{eq:energy2}
%\end{eqnarray}
where we have set $\epsilon_1= \frac{\rho_{\rm min}}{4}$, 
$\epsilon_2= (\rho_{\min}  \alpha_{\min })^{1/2} / (2 \rho_{\max}^{1/2})$,
and the constants $C_1$, $C_2$ depend on $u_2$, $\mu$, $\rho_{\rm max}$,
$\rho_{\rm min}$, $\alpha_{\rm max}$, $\alpha_{\rm min}$ and $q$. 
Gr{\"o}nwall's inequality for $\nabla v$ implies that $\|\nabla v (t) \|_2^2$ is bounded
in terms of  $C_1 \|\alpha_{h_1} - \alpha_{h_2}\|_{L^2(\Omega)}^{2/q'} +
 C_2 \| \alpha_{h_1}  -  \alpha_{h_2} \|_{L^2(\partial \Omega^-)}^{2/q'} $ for $t\in[0,T]$.
 Then, (\ref{eq:energy2}) implies that $\|v_t (t) \|_2^2$, $\|v (t) \|_{H^1}^2$ and 
 $\|v (t) \|_{L^2(\partial \Omega^+)}$ (thanks to trace inequalities) too.
 \end{pf}

 \begin{corollary} \label{thm:A3}
 In our particular setting, the conclusion of Theorem 2
 remains true replacing the norms of $\alpha_{h_1}-\alpha_{h_2}$ by 
 $\|h_1-h_2\|_\infty$.
 \end{corollary}
 \begin{pf} Notice that
\begin{equation}
    \begin{aligned}
        \int_{\Omega} | \alpha_{h_1} -\alpha_{h_2} |^2 d \mathbf x &\leq 2 \alpha_{\rm max}^2 
\int_a^b dx \int_{{\rm min}(h_1(x),h_2(x))}^{{\rm max}(h_1(x),h_2(x))} dy \\
        &\leq 2 \alpha_{\rm max}^2 \int_a^b |{\rm max}(h_1(x),h_2(x)) - {\rm min}(h_1(x),h_2(x))| dx \\
        &\leq 2 \alpha_{\rm max}^2 (b-a) \| h_1-h_2 \|_\infty.
    \end{aligned}
\end{equation}
 
% \begin{eqnarray*}
%\int_{\Omega} | \alpha_{h_1} -\alpha_{h_2} |^2 d \mathbf x \leq 2 \alpha_{\rm max}^2 
%\int_a^b dx \int_{{\rm min}(h_1(x),h_2(x))}^{{\rm max}(h_1(x),h_2(x))} dy  
%\\
%\leq 2 \alpha_{\rm max}^2 \int_a^b |{\rm max}(h_1(x),h_2(x)) - {\rm min}(h_1(x),h_2(x))| dx
%\leq 2 \alpha_{\rm max}^2 (b-a) \| h_1-h_2 \|_\infty. 
% \end{eqnarray*}
c \end{pf}

\begin{rmk} \label{thm:A4}
Combining Theorem 2 with Corollary 3 in our case we can conclude that 
\[
    \|u_{h_1} - u_{h_2} \|_{L^2(\partial \Omega^+)}  \leq C \| h_1 - h_2 \|_{L^2},
\]
for some $C>0$.
\end{rmk}

\printcredits

\bibliographystyle{cas-model2-names}

\bibliography{cas-refs}

@article{adams,
  title={{S}obolev {S}paces},
  author={Adams, Robert A.},
  publisher = {Academic Press},
  year = {1975},
}

@article{lionsmagenes,
  title={Problemes aux limites non-homogenes et applications},
  author={Lions, Jacques Louis and Magenes, Enrico},
  publisher = {Dunod},
  volume={1},
  year = {1968},
}

@article{grisvard,
  title={ Elliptic Problems in Nonsmooth Domains},
  author={Grisvard, Pierre},
  publisher = {Pitman},
  year = {1985},
}

@article{brezis,
  title={Analyse fonctionnelle},
  author={Brezis, Haim},
  publisher = {Masson},
  year = {1987},
}

@article{tromp2020earth,
  title={Seismic wavefield imaging of Earth’s interior across scales},
  author={Tromp, Jeroen},
  journal={Nature Reviews Earth \& Environment},
  volume={1},
  number={},
  pages={40-53},
  year={2020},
  publisher={}
}

@article{tsogka2002time,
  title={Time reversal through a solid--liquid interface and super-resolution},
  author={Tsogka, Chrysoula and Papanicolaou, George C},
  journal={Inverse problems},
  volume={18},
  number={6},
  pages={1639},
  year={2002},
  publisher={IOP Publishing}
}

@article{borcea2021reduced,
  title={Reduced order model approach for imaging with waves},
  author={Borcea, Liliana and Garnier, Josselin and Mamonov, Alexander V and Zimmerling, J{\"o}rn},
  journal={Inverse Problems},
  volume={38},
  number={2},
  pages={025004},
  year={2021},
  publisher={IOP Publishing}
}

@book{Quarteroni2014,
  title = {Numerical Models for Differential Problems},
  ISBN = {9788847055223},
  url = {http://dx.doi.org/10.1007/978-88-470-5522-3},
  DOI = {10.1007/978-88-470-5522-3},
  publisher = {Springer Milan},
  author = {Quarteroni,  Alfio},
  year = {2014}
}

@book{Edsberg2015-rk,
  title = {{I}ntroduction {T}o {C}omputation and {M}odeling {F}or {D}ifferential {E}quations},
  author = {Edsberg, Lennart},
  publisher = {John Wiley \& Sons},
  edition = {2},
  year = {2015},
}

@book{hairer2006structure,
  title = {{S}tructure-{P}reserving {A}lgorithms {F}or {O}rdinary {D}ifferential {E}quations},
  author = {Hairer, Ernst and Lubich, Christian and Wanner, Gerhard},
  year = {2006},
  publisher = {Springer-Verlag Berlin}
}

@book{gatto2022mathematical,
  title={{M}athematical {F}oundations of {F}inite {E}lements and {I}terative {S}olvers},
  author={Gatto, Paolo},
  year={2022},
  publisher={SIAM}
}

@book{ibragimov2012gaussian,
  title={{G}aussian {R}andom {P}rocesses},
  author={Ibragimov, Ildar Abdulovich and Rozanov, Yurii Antol'evich},
  volume={9},
  year={2012},
  publisher={Springer Science \& Business Media}
}

@book{dudley2018real,
  title={{R}eal {A}nalysis and {P}robability},
  author={Dudley, Richard M},
  year={2018},
  publisher={CRC Press}
}

@article{dunlop2017hierarchical,
  title={Hierarchical {B}ayesian level set inversion},
  author={Dunlop, Matthew M and Iglesias, Marco A and Stuart, Andrew M},
  journal={Statistics and Computing},
  volume={27},
  pages={1555--1584},
  year={2017},
  publisher={Springer}
}

@Inbook{Dashti2017,
    author={Dashti, Masoumeh and Stuart, Andrew M.},
    title={The {B}ayesian Approach to Inverse Problems},
    bookTitle={{H}andbook of {U}ncertainty {Q}uantification},
    year={2017},
    publisher={Springer International Publishing},
    pages={311--428},
    doi={10.1007/978-3-319-12385-1_7},
}

@book{tolstov2012fourier,
  title={{F}ourier {S}eries},
  author={Tolstov, Georgi P},
  year={2012},
  publisher={Courier Corporation}
}

@book{briggs1995dft,
  title={The DFT: an owner's manual for the discrete Fourier transform},
  author={Briggs, William L and Henson, Van Emden},
  year={1995},
  publisher={SIAM}
}

@article{carpio2023shear,
    doi = {10.1088/1361-6420/acd5f8},
    url = {https://dx.doi.org/10.1088/1361-6420/acd5f8},
    year = {2023},
    month = {jun},
    publisher = {IOP Publishing},
    volume = {39},
    number = {7},
    pages = {075007},
    author = {Ana Carpio and Elena Cebrián and Andrea Gutiérrez},
    title = {{O}bject based {B}ayesian full-waveform inversion for shear elastography},
    journal = {Inverse Problems},
}

@article{maboudi2024,
  title={Inferring Object Boundaries and Their Roughness with Uncertainty Quantification},
  author={Afkham, Babak Maboudie and Riis, Nicolai Andr{\'e} Brogaard and Dong, Yiqiu and Hansen, Per Christian},
  journal={Journal of Mathematical Imaging and Vision},
  pages={1--16},
  year={2024},
  publisher={Springer},
  doi={https://doi.org/10.1007/s10851-024-01207-9}
}

@article{roininen2014whittle,
  title={Whittle-{M}at{\'e}rn priors for {B}ayesian statistical inversion with applications in electrical impedance tomography.},
  author={Roininen, Lassi and Huttunen, Janne MJ and Lasanen, Sari},
  journal={Inverse Problems \& Imaging},
  volume={8},
  number={2},
  year={2014}
}

@book{schervish2012theory,
  title={{T}heory of {S}tatistics},
  author={Schervish, Mark J},
  year={2012},
  publisher={Springer Science \& Business Media}
}

@article{iglesias2016bayesian,
  title={A {B}ayesian level set method for geometric inverse problems},
  author={Iglesias, Marco A and Lu, Yulong and Stuart, Andrew M},
  journal={Interfaces and free boundaries},
  volume={18},
  number={2},
  pages={181--217},
  year={2016}
}

@book{le2000asymptotics,
  title={{A}symptotics in {S}tatistics: {S}ome {B}asic {C}oncepts},
  author={Le Cam, Lucien Marie and Yang, Grace Lo},
  year={2000},
  publisher={Springer Science \& Business Media}
}

@article{e6c3f9b1-115e-37a9-b4d4-010ab9364de8,
  ISSN = {08834237, 21688745},
  URL = {http://www.jstor.org/stable/43288425},
  author = {S. L. Cotter and G. O. Roberts and A. M. Stuart and D. White},
  journal = {Statistical Science},
  number = {3},
  pages = {424--446},
  publisher = {Institute of Mathematical Statistics},
  title = {{MCMC} Methods for Functions: Modifying Old Algorithms to Make Them Faster},
  urldate = {2024-03-20},
  volume = {28},
  year = {2013}
}

@article{feldman1958equivalence,
  title={{E}quivalence and perpendicularity of {G}aussian processes},
  author={Feldman, Jacob},
  journal={Pacific J. Math},
  volume={8},
  number={4},
  pages={699--708},
  year={1958}
}

@book{kaipio2006statistical,
  title={{S}tatistical and {C}omputational {I}nverse {P}roblems},
  author={Kaipio, Jari and Somersalo, Erkki},
  volume={160},
  year={2006},
  publisher={Springer Science \& Business Media}
}

@article{karamehmedovic2013efficient,
  title={An efficient rough-interface scattering model for embedded nano-structures},
  author={Karamehmedovi{\'c}, Mirza and Hansen, Poul-Erik and Wriedt, Thomas},
  journal={Thin Solid Films},
  volume={541},
  pages={51--56},
  year={2013},
  publisher={Elsevier}
}

@article{schroder2011modeling,
  title={Modeling of light scattering in different regimes of surface roughness},
  author={Schr{\"o}der, Sven and Duparr{\'e}, Angela and Coriand, Luisa and T{\"u}nnermann, Andreas and Penalver, Dayana H and Harvey, James E},
  journal={Optics express},
  volume={19},
  number={10},
  pages={9820--9835},
  year={2011},
  publisher={Optica Publishing Group}
}

@article{sym13091702,
  author = {Carpio, Ana and Rapún, María-Luisa},
  title = {Multifrequency Topological Derivative Approach to Inverse Scattering Problems in Attenuating Media},
  journal = {Symmetry},
  volume = {13},
  year = {2021},
  number = {9},
  issn = {2073-8994},
  doi = {10.3390/sym13091702}
}

@misc{sherlock2009optimal,
  title={Optimal scaling of the random walk Metropolis on elliptically symmetric unimodal targets},
  author={Sherlock, Chris and Roberts, Gareth},
  year={2009}
}

@article{afkham2024bayesian,
  title={A Bayesian approach for consistent reconstruction of inclusions},
  author={Afkham, Babak Maboudi and Knudsen, Kim and Rasmussen, Aksel Kaastrup and Tarvainen, Tanja},
  journal={Inverse Problems},
  volume={40},
  number={4},
  pages={045004},
  year={2024},
  publisher={IOP Publishing}
}

@book{robert1999monte,
  title={{M}onte {C}arlo {S}tatistical {M}ethods},
  author={Robert, Christian P and Casella, George and Casella, George},
  volume={2},
  year={1999},
  publisher={Springer}
}

@article{fan1996study,
  title={A study of variable bandwidth selection for local polynomial regression},
  author={Fan, Jianqing and Gijbels, Ir{\`e}ne and Hu, Tien-Chung and Huang, Li-Shan},
  journal={Statistica Sinica},
  pages={113--127},
  year={1996},
  publisher={JSTOR}
}

@article{arviz_2019,
  doi = {10.21105/joss.01143},
  url = {https://doi.org/10.21105/joss.01143},
  year = {2019},
  publisher = {The Open Journal},
  volume = {4},
  number = {33},
  pages = {1143},
  author = {Ravin Kumar and Colin Carroll and Ari Hartikainen and Osvaldo Martin},
  title = {ArviZ a unified library for exploratory analysis of Bayesian models in Python},
  journal = {Journal of Open Source Software}
}

@book{jackson2007high,
  title={{H}igh-frequency seafloor acoustics},
  author={Jackson, Darrell and Richardson, Michael},
  year={2007},
  publisher={Springer Science \& Business Media}
}

@article{anderson2008acoustic,
  title={Acoustic seabed classification: current practice and future directions},
  author={Anderson, John T and Van Holliday, D and Kloser, Rudy and Reid, Dave G and Simard, Yvan},
  journal={ICES Journal of Marine Science},
  volume={65},
  number={6},
  pages={1004--1011},
  year={2008},
  publisher={Oxford University Press}
}

@book{hansen2021computed,
  title={{C}omputed {T}omography:\ {A}lgorithms, {I}nsight, and {J}ust {E}nough {T}heory},
  editor={Hansen, Per Christian and J{\o}rgensen, Jakob and Lionheart, William RB},
  year={2021},
  publisher={SIAM}
}

@article{frederick2020seabed,
  title={Seabed classification using physics-based modeling and machine learning},
  author={Frederick, Christina and Villar, Soledad and Michalopoulou, Zoi-Heleni},
  journal={The Journal of the Acoustical Society of America},
  volume={148},
  number={2},
  pages={859--872},
  year={2020},
  publisher={AIP Publishing}
}

@article{borcea2023data,
  title={When data driven reduced order modeling meets full waveform inversion},
  author={Borcea, Liliana and Garnier, Josselin and Mamonov, Alexander V and Zimmerling, J{\"o}rn},
  journal={arXiv preprint arXiv:2302.05988},
  year={2023}
}

@book{koller2009probabilistic,
  title={{P}robabilistic {G}raphical {M}odels: {P}rinciples and {T}echniques},
  author={Koller, Daphne and Friedman, Nir},
  year={2009},
  publisher={MIT press}
}

@article{suuronen2022cauchy,
  title={{C}auchy {M}arkov random field priors for {B}ayesian inversion},
  author={Suuronen, Jarkko and Chada, Neil K and Roininen, Lassi},
  journal={Statistics and computing},
  volume={32},
  number={2},
  pages={33},
  year={2022},
  publisher={Springer}
}

@article{bardsley2012laplace,
  title={Laplace-distributed increments, the {L}aplace prior, and edge-preserving regularization},
  author={Bardsley, Johnathan M},
  journal={Journal of Inverse and Ill-Posed Problems},
  volume={20},
  number={3},
  pages={271--285},
  year={2012},
  publisher={Walter de Gruyter GmbH \& Co. KG}
}

@article{stuart2010inverse,
  title={Inverse problems: a {B}ayesian perspective},
  author={Stuart, Andrew M},
  journal={Acta numerica},
  volume={19},
  pages={451--559},
  year={2010},
  publisher={Cambridge University Press}
}

@book{bjorno2017applied,
  title={{A}pplied {U}nderwater {A}coustics},
  author={Bj{\o}rn{\o}, Leif},
  year={2017},
  publisher={Elsevier}
}

@inproceedings{valentine2005classification,
  title={Classification of Marine Sublittoral Habitats, with Application to the Northeastern North Application to the Northeastern North {A}merica Region},
  author={Valentine, Page C and Todd, Brian J and Kostylev, Vladimir E},
  booktitle={American Fisheries Society Symposium},
  volume={41},
  pages={183--200},
  year={2005}
}

@book{adams2003sobolev,
  title={Sobolev spaces},
  author={Adams, Robert A and Fournier, John JF},
  volume={140},
  year={2003},
  publisher={Elsevier}
}

@book{landau2012theory,
  title={Theory of elasticity: volume 7},
  author={Landau, Lev Davidovich and Pitaevskii, LP and Kosevich, Arnolʹd Markovich and Lifshitz, Evgenii Mikhailovich},
  volume={7},
  year={2012},
  publisher={Elsevier}
}

@article{ABUGATTAS2025129453,
title = {Quantifying uncertainty in inverse scattering problems set in layered environments},
journal = {Applied Mathematics and Computation},
volume = {500},
pages = {129453},
year = {2025},
issn = {0096-3003},
doi = {https://doi.org/10.1016/j.amc.2025.129453},
author = {Carolina Abugattas and Ana Carpio and Elena Cebrián and Gerardo Oleaga},
}

@article{WANG2025104383,
title = {{CUG-STCN}: A seabed topography classification framework based on knowledge graph-guided vision mamba network},
journal = {International Journal of Applied Earth Observation and Geoinformation},
volume = {136},
pages = {104383},
year = {2025},
issn = {1569-8432},
doi = {https://doi.org/10.1016/j.jag.2025.104383},
url = {https://www.sciencedirect.com/science/article/pii/S1569843225000305},
author = {Haoyi Wang and Weitao Chen and Xianju Li and Qianyong Liang and Xuwen Qin and Jun Li}
}

@article{lurton2015backscatter,
  title={Backscatter measurements by seafloor-mapping sonars},
  author={Lurton, Xavier and Lamarche, Geoffroy and Brown, Craig and Lucieer, Vanessa and RIce, Glen and Schimel, Alexandre and Weber, T},
  journal={Guidelines and recommendations},
  volume={200},
  year={2015},
  publisher={GeoHab Brest, France}
}

@inproceedings{clarke2015multispectral,
  title={Multispectral acoustic backscatter from multibeam, improved classification potential},
  author={Clarke, John E Hughes},
  booktitle={Proceedings of the United States Hydrographic Conference, San Diego, CA, USA},
  pages={16--19},
  year={2015}
}

@article{jia2021underwater,
  title={Underwater reverberation suppression based on non-negative matrix factorisation},
  author={Jia, Hongjian and Li, Xiukun},
  journal={Journal of Sound and Vibration},
  volume={506},
  pages={116166},
  year={2021},
  publisher={Elsevier}
}

@article{jenserud2015measurements,
  title={Measurements and modeling of effects of out-of-plane reverberation on the power delay profile for underwater acoustic channels},
  author={Jenserud, Trond and Ivansson, Sven},
  journal={IEEE Journal of Oceanic Engineering},
  volume={40},
  number={4},
  pages={807--821},
  year={2015},
  publisher={IEEE}
}

@article{fofonoff1983algorithms,
  title={Algorithms for the computation of fundamental properties of seawater.},
  author={Fofonoff, Nicholas Paul and Millard Jr, RC},
  year={1983},
  publisher={Unesco}
}

@article{roberts1998optimal,
  title={Optimal scaling of discrete approximations to {L}angevin diffusions},
  author={Roberts, Gareth O and Rosenthal, Jeffrey S},
  journal={Journal of the Royal Statistical Society: Series B (Statistical Methodology)},
  volume={60},
  number={1},
  pages={255--268},
  year={1998},
  publisher={Wiley Online Library}
}

@article{duane1987hybrid,
  title={Hybrid monte carlo},
  author={Duane, Simon and Kennedy, Anthony D and Pendleton, Brian J and Roweth, Duncan},
  journal={Physics letters B},
  volume={195},
  number={2},
  pages={216--222},
  year={1987},
  publisher={Elsevier}
}

@article{coullon2021ensemble,
  title={Ensemble sampler for infinite-dimensional inverse problems},
  author={Coullon, Jeremie and Webber, Robert J},
  journal={Statistics and Computing},
  volume={31},
  number={3},
  pages={28},
  year={2021},
  publisher={Springer}
}

@article{goodman2010ensemble,
  title={Ensemble samplers with affine invariance},
  author={Goodman, Jonathan and Weare, Jonathan},
  journal={Communications in applied mathematics and computational science},
  volume={5},
  number={1},
  pages={65--80},
  year={2010},
  publisher={Mathematical Sciences Publishers}
}

@article{emcee,
   author = {{Foreman-Mackey}, D. and {Hogg}, D.~W. and {Lang}, D. and {Goodman}, J.},
    title = {emcee: The MCMC Hammer},
  journal = {PASP},
     year = 2013,
   volume = 125,
    pages = {306-312},
   eprint = {1202.3665},
      doi = {10.1086/670067}
}

@article{yamamoto,
author= {Yamamoto, Tokuo and Torii, Tsuyoshi},
title= {Seabed shear modulus profile inversion using surface gravity (water) wave-induced bottom motion},
journal= {Geophysical Journal International},
volume= {85},
number ={2},
year={1986}, 
pages={413--431},
doi = {10.1111/j.1365-246X.1986.tb04521.x}
}

@article{waveguide,
title= {Passive geoacoustic inversion using waveguide characteristic impedance: a sensitivity study Free},
author= {Ren, Qunyan and Hermand, Jean-Pierre},
journal= {Proc. Mtgs. Acoust.},
volume= {17}, 
pages={070095},
year={2012},
doi = {10.1121/1.4795530}
}

@article{multistep,
author= {Yan, X. and Yang, C. and Zhang, T. et al.},
title= {Multi-step inversion of seabed geoacoustic and scattering parameters using joint propagation and reverberation data},
journal= {Mar Geophys Res},
volume= {46}, 
number={13},
year={2025},
doi = {10.1007/s11001-025-09575-6}
}

@incollection{bayesian1,
  author    = {Dosso, S. E. and Holland, C. W.},
  title     = {Bayesian inversion of seabed reflection data},
  booktitle = {Acoustic Sensing Techniques for the Shallow Water Environment},
  editor    = {Caiti, A. and Chapman, N. R. and Hermand, J.-P. and Jesus, S. M.},
  publisher = {Springer},
  address   = {Dordrecht},
  year      = {2006},
  doi       = {10.1007/978-1-4020-4386-4_2}
}

@article{bayesian2,
  title   = {Bayesian geoacoustic parameter inversion for multi-layer seabed in shallow sea using underwater acoustic fields},
  author  = {Xue, Yangyang and Zhu, Hanhao and Wang, Xiaohan and Zheng, Guangxue and Liu, Xu and Wang, Jiahui},
  journal = {Frontiers in Marine Science},
  volume  = {10},
  year    = {2023},
  doi     = {10.3389/fmars.2023.1058542}
}

@article{tomography,
  title   = {From geoacoustic inversion to seabed tomography using a distributed network of sources and receivers},
  author  = {Bonnel, Julien and Vardi, Ariel and Leonard, John and Dosso, Stan},
  journal = {The Journal of the Acoustical Society of America},
  volume  = {157},
  number  = {4},
  pages   = {A196--A196},
  year    = {2025},
  doi     = {10.1121/10.0037871}
}

@book{jensen2011computational,
  title={Computational ocean acoustics},
  author={Jensen, Finn B and Kuperman, William A and Porter, Michael B and Schmidt, Henrik and Tolstoy, Alexandra},
  volume={2011},
  year={2011},
  publisher={Springer}
}

@article{gelbaum2012white,
  title={White noise representation of Gaussian random fields},
  author={Gelbaum, Zachary},
  journal={arXiv preprint arXiv:1201.5635},
  year={2012}
}

@book{brigham1988fast,
  title={The fast Fourier transform and its applications},
  author={Brigham, E Oran},
  year={1988},
  publisher={Prentice-Hall, Inc.}
}



\end{document}